\documentclass[a4paper,11pt]{article}
\usepackage[T1]{fontenc} % if needed
\usepackage[section]{placeins}

\usepackage{physics}

\usepackage{mathtools}
\usepackage{amssymb}
\usepackage{amsthm}
\usepackage{array}
\usepackage{authblk}
\usepackage{fancyhdr}
\usepackage{hyperref}
\usepackage{float}
\usepackage{color} 
\usepackage{tikz}
\usetikzlibrary{tqft}
\usepackage[justification=justified,singlelinecheck=false]{subcaption}

\def\mC{ \mathbb{C} } 
\def\cZ{ \mathcal{Z}}

\DeclareMathOperator{\Span}{Span}

\usepackage[backend=bibtex,style=phys,eprint=true,doi=false,isbn=false,url=false,sorting=none]{biblatex}
\bibliography{bib.bib}

\theoremstyle{definition}
\newtheorem{theorem}{Theorem}
\newtheorem{proposition}{Proposition}
\newtheorem{definition}{Definition}
\newtheorem{conjecture}{Conjecture}
\newtheorem{open}{Open Question}

\title{ \bf 
Row and column detection complexities of character tables 
}

\renewcommand*{\Authsep}{, }
\renewcommand*{\Authand}{, }%<---------------remove and
\renewcommand*{\Authands}{, }
\renewcommand*{\Affilfont}{\small}

\author{{\bf \small  Adrian Padellaro}\thanks{apadellaro@physik.uni-bielefeld.de}}
\affil{Faculty of Physics, Bielefeld~University, \authorcr PO-Box 100131, D-33501 Bielefeld, Germany}
\author{{\bf  \small Sanjaye Ramgoolam}\thanks{s.ramgoolam@qmul.ac.uk}}
\affil{Centre for Theoretical Physics, Department of Physics and Astronomy, Queen~Mary University of London, London E1 4NS, United Kingdom}
\author{{\bf \small  Rak-Kyeong Seong}\thanks{seong@unist.ac.kr}}
\affil{Department of Mathematical Sciences, and Department of Physics, Ulsan~National~Institute~of~Science~and~Technology, \authorcr 50 UNIST-gil, Ulsan 44919, South Korea}

\date{}

\begin{document} 
\maketitle
\thispagestyle{fancy}
\begin{abstract}
Character tables of finite groups and closely related commutative algebras have been investigated recently using new perspectives arising from the AdS/CFT correspondence and low-dimensional topological quantum field theories.  Two important elements in these new  perspectives are physically motivated definitions of quantum complexity for the algebras and a notion of row-column duality. These elements are encoded in properties of  the 
character table of a group $G$ and the associated algebras, notably the centre of the group algebra  and the fusion algebra of  irreducible representations of the group. 
Motivated by these developments, we define row and column versions of detection  complexities for character tables, and investigate the relation between these complexities under the exchange of rows and columns.
We observe regularities that arise in the statistical averages over small character tables and propose corresponding conjectures for arbitrarily large character tables.

\end{abstract}
\flushbottom
\pagestyle{plain}

\newpage

\tableofcontents

\section{Introduction}

The AdS/CFT correspondence \cite{Maldacena:1997re,witten,gkp} and topological quantum field theory  have been used to produce some new perspectives and investigations on centres of group algebras \cite{KempRam,JosephSRKronecker,JosephSRProj,ICFDTS,SRES,RowCol}. 
In the canonical example of AdS/CFT relating $U(N)$ super-Yang-Mills gauge theory in  four dimensions to string theory in $AdS_5 \times S^5$, an important role is played  by half-BPS operators in the SYM theory and their labelling by Young diagrams \cite{Corley:2001zk}. Young diagrams with $n$ boxes, and having columns of height no larger than $N$, are associated with half-BPS operators of dimension $n$. For $ n < N $, the cut-off on the column heights is immaterial and AdS/CFT leads to questions about centres of group algebras of symmetric groups $ \mC[S_n]$.

\vskip.2cm

The center of $\mathbb{C}[S_n]$, denoted $ \cZ (\mathbb{C}[S_n] ) $, has a distinguished basis labelled by conjugacy classes of $S_n$, and recently an integer $k_*(n)$ was identified as important for understanding the algebraic structure of the center \cite{KempRam}.
The integer $k_*(n)$ gives the size of a subset of conjugacy class labels that correspond to non-linearly generating basis elements of the center of the group algebra.
Subsets of non-linearly generating basis elements correspond to subsets of the columns of the character table, which can be used to distinguish the rows which are labelled by irreducible representations or equivalently Young diagrams in the case of $S_n$ \cite{KempRam,SRES}.
The cardinalities of these subsets give measures of complexity associated with quantum projector detection questions \cite{JosephSRProj} motivated by toy models of black hole information loss in AdS/CFT \cite{IIloss}. These minimal generating subsets are also used in construction algorithms for finite $N$ integer bases of multi-matrix invariants, which has applications to AdS/CFT \cite{PRS2024}.

\vskip.2cm 

Centres of finite group algebras also form the Hilbert spaces of Dijkgraaf-Witten {(DW)} two-dimensional topological field theories (TQFTs) \cite{Dijkgraaf:1989pz} associated with the group $G$.
The simplest form of these theories counts equivalence classes of principal $G$ bundles on two-dimensional surfaces. They have a simple topological lattice realisation \cite{FHK}, which uses a sum over group elements on edges of the lattice.
There is a delta-function weight associated with faces, which imposes the condition that the product of group elements around the face is the identity.

\vskip.2cm 

Observables in $S_n$ DW theory also appear in the large $N$ expansion of two-dimensional Yang-Mills theory  \cite{Gross:1993hu} and play an important role in the stringy interpretation of this large $N$ expansion \cite{DAdda:2002fcu}. 
Recent studies based on string theory ideas combined with DW theory gave physical constructions of representation theoretic numbers \cite{ICFDTS}.
This was extended to projective representations in \cite{SRES}, and several integrality results for partial sums of characters along columns of  character tables of finite groups were given. 
The integrality results were extended in \cite{RowCol} by introducing a row-column dual TQFT based on the fusion algebra $R(G)$ of a finite group $G$.

\vskip.2cm 

These different strands of development motivate a study of finite group character tables, from perspectives of complexity based on AdS/CFT and topological quantum field theories.
An important focus of the present study is the size of subsets of conjugacy classes, which non-linearly generate the centres of group algebras.
We initiate here an empirical study for general groups and seek minimum generating sets of conjugacy classes, searched among the full set of conjugacy classes.
The characters of these minimal sets serve to distinguish the complete set of irreducible representations.
This is a generalisation of the question considered
in \cite{KempRam} where the focus was on subsets consisting of conjugacy classes in symmetric groups where the permutations have a single non-trivial cycle and 
the remaining cycles are of length one.

\vskip.2cm 

Following the notion of row-column duality introduced in the context of TQFT constructions of integer partial sums of characters in \cite{PRS2024}, we introduce here a row-column dual question for minimal generating subsets of irreducible representations for the fusion algebras of finite groups.
We show, by extending the reasoning of \cite{KempRam,SRES}, that these correspond to minimal subsets of the rows of a given character table, which serve to distinguish all the columns.

\vskip.2cm

Our investigation is aligned with the recent theme in string theory research of looking at mathematically defined data of physical interest, using an
interdisciplinary exchange of ideas between mathematics, physics and machine learning \cite{He:2017aed,Krefl:2017yox, Ruehle:2017mzq, Carifio:2017bov, He:2017set, Jejjala:2019kio, Brodie:2019dfx, Ashmore:2019wzb, Gukov:2020qaj, Hashimoto:2018ftp, Cole:2018emh, Halverson:2020trp}. 
Data-driven machine learning algorithms are being harnessed in order to analyze vast datasets available in mathematics and relevant in string theory.
The use of machine learning is, for example, redefining our understanding of Calabi-Yau manifolds in string theory \cite{Bull:2018uow, Choi:2023rqg, Seong:2023njx},
leading to discoveries of new phenomena in number theory \cite{He:2022pqn}, and revealing new hidden structures in BPS spectra and corresponding 3-manifolds \cite{Gukov:2024opc}.
The idea that mathematical meaning is yet to be discovered in large datasets available to us in string theory, 
and the prospect of discovering new physical implications for string theory in the future through machine learning is extremely compelling and underscores the timeliness of our work.

\vskip.2cm 

The paper is organised as follows. The main sections are preceded by a page of nomenclature, detailing our notation. In Section \ref{sec: tqft}, we review the row-column duality for two-dimensional TQFTs based on finite groups.
The first Subsection \ref{sec: review tqft} gives a brief review of the most relevant aspects of the connection between two-dimensional topological quantum field theory and commutative Frobenius algebras.
We define the notion of a combinatorial basis for  the commutative Frobenius algebras of interest here. This is used to define a circle generator complexity and a circle-and-handle generator complexity for semi-simple Frobenius algebras having such combinatorial bases.
In Subsection \ref{sec: DW}, we review Dijkgraaf-Witten theories, which are based on the class algebra of a finite group $G$ and have a combinatorial basis labelled by conjugacy classes of $G$.
Subsection \ref{sec: class complexity} specialises the definitions of generator complexity to DW theories and introduces an average class size associated to the  circle generator complexity. 
Subsection \ref{sec: fusion} describes the row-column duals of Dijkgraaf-Witten theories, the so-called fusion TQFTs, which are based on the fusion algebra of a finite group and has a combinatorial basis labelled by irreducible representations of $G$.
Subsection \ref{sec: fusion complexity} specialises the definition of the generator complexities to Fusion TQFTs and introduces an average square dimension associated to the circle generator complexity. 
In section \ref{sec: stats}, we initiate the empirical study of the above generator complexities and associated averages. In particular, we study these features for  character tables up to size $n=30$. We make several observations about the statistical behaviour of these features and conjecture that some of the observations hold in general.

\newpage
\paragraph{Nomenclature}
\begin{align*}
G 
\quad &: \quad \text{a finite group}
\\
|G|
\quad &: \quad
\text{order of the group $G$}
\\
\mathrm{Cl}(G) 
\quad &: \quad
\text{set of conjugacy classes of $G$}
\\
|C|
\quad &: \quad
\text{size of the conjugacy class $C \in \mathrm{Cl}(G)$}
\\
\mathrm{Irr}(G)
\quad &: \quad
\text{set of  irreducible representations of $G$}
\\
d_R
\quad &: \quad
\text{dimension of the representation $R \in \mathrm{Irr}(G)$}
\\
\mathbb{C}[G]
\quad &: \quad
\text{group algebra of $G$ over the field of complex numbers $ \mathbb{C}$}
\\
\mathcal{Z}(\mathbb{C}[G])
\quad &: \quad
\parbox[t]{0.8\linewidth}{class algebra/center of group algebra of $G$, defined by multiplication of sums of group elements over conjugacy classes}
\\
R(G)
\quad &: \quad
\parbox[t]{0.8\linewidth}{fusion algebra of $G$, defined using the decomposition into irreducibles of the  tensor product of irreducible representations}
\\
\chi^R_C
\quad &: \quad
\parbox[t]{0.8\linewidth}{value of irreducible character $\chi^R$ on some $g \in C$ for a conjugacy class $C$}
\\
N_{\text{cls}}(G)
\quad &: \quad
\parbox[t]{0.8\linewidth}{circle generator complexity of class algebra of $G$}
\\
N^{\text{ch}}_{\text{cls}}(G)
\quad &: \quad
\parbox[t]{0.8\linewidth}{circle-and-handle generator complexity of class algebra of $G$}
\\
N_{\text{fus}}(G)
\quad &: \quad
\parbox[t]{0.8\linewidth}{circle generator complexity of fusion algebra of $G$}
\\
N^{\text{ch}}_{\text{fus}}(G)
\quad &: \quad
\parbox[t]{0.8\linewidth}{circle-and-handle generator complexity of fusion algebra of $G$}
\\
\mathcal{N}(G)
\quad &: \quad
\parbox[t]{0.8\linewidth}{average size of conjugacy classes of $G$, also average squared dimension of irreducible representations of $G$.}
\\
\mathcal{C}_{\text{gens}}(G)
\quad &: \quad
\text{average class size over generating subsets for class algebra of $G$}
\\
\mathcal{R}_{\text{gens}}(G)
\quad &: \quad
\text{average squared dimension over generating subsets for $R(G)$}
\end{align*}
\newpage

\section{Row-Column dual 2D TQFTs based on groups}
\label{sec: tqft}
In this section, we review some aspects of two-dimensional topological quantum field theory (2D TQFT) and their relation to commutative Frobenius algebras.
We also give some general theorems and definitions that are used in subsequent subsections where we focus on TQFTs based on finite groups.
This includes two definitions of complexity associated with TQFTs, which correspond to commutative semi-simple Frobenius algebras with combinatorial bases.
In particular, we review Dijkgraaf-Witten TQFTs and fusion TQFTs, which are based on the class algebra and fusion algebra of a finite group $G$, respectively.
They are used to define a row-column duality for the integrality properties of partial sums along rows and columns of the character table of $G$ \cite{RowCol}.
\\

\subsection{2D TQFTs, Frobenius algebras and complexity}
\label{sec: review tqft}

The Atiyah-Segal  \cite{Atiyah1988,Segal1999} axioms 
regard 2D TQFT as a functor from the category of two-dimensional oriented cobordisms (with objects given by disjoint unions of circles, and morphisms given by cobordisms between them) to the category of vector spaces.
An important result is that  a 2D TQFT is uniquely determined 
 by  a  finite-dimensional commutative Frobenius algebra (see \cite{JKock} for a review of the subject).

This subsection reviews the most relevant aspects of 2D TQFTs and commutative Frobenius algebras.
First, we give the definition of a commutative Frobenius algebra.
\begin{definition}[Commutative Frobenius algebra]
	A commutative Frobenius algebra $A$ over $\mathbb{C}$ is defined by
	\begin{enumerate}
		\item A finite-dimensional commutative and associative algebra $A$, which is
		a finite-dimensional vector space together with a commutative and associative product $\mu: A \otimes A \rightarrow A$.
		\item A unit $\eta: \mathbb{C} \rightarrow A$ such that $\mu(\eta(c),a) = \mu(a, \eta(c)) = ca$ for all $c \in \mathbb{C}$ and $a \in A$.
		\item A linear function $\varepsilon: A \rightarrow \mathbb{C}$, called the co-unit, with the property that the bilinear form $g = \varepsilon \circ \mu$ is non-degenerate.
		In other words, if $\{e_i\}_{i=1}^K$ is a basis for $A$, then the matrix
		\begin{equation}
			g_{ij} = \varepsilon(\mu(e_i,e_j)) \label{eq: frob form}
		\end{equation}
		is invertible.
	\end{enumerate}
	We often omit the symbol $\mu$ and write $\mu(a,b) = ab$ for the product of $a,b \in A$. 
\end{definition}

According to the Atiyah-Segal axioms, the above objects are associated with a set of surfaces that form the building blocks for general surfaces. 
The product is associated with the pair of pants surface, the unit is a cup and the co-unit is a cap, as illustrated below,
\begin{equation}
	\begin{aligned}
	&\mu: A \otimes A \rightarrow A &&\longleftrightarrow 
	&&\vcenter{
	\hbox{
	\begin{tikzpicture}[tqft/cobordism/.style={draw},
		tqft/view from=outgoing,
		tqft/boundary separation=20pt,
		tqft/cobordism height=25pt,
		tqft/circle x radius=5pt,
		tqft/circle y radius=2pt,
		tqft/every boundary component/.style={draw,rotate=0}]
   		\pic[tqft,
	   		incoming boundary components=2,
	   		outgoing boundary components=1,
			offset=.5,
	   		rotate=90,
			name=b,anchor={(0,0)}];
		\node at ([xshift=12pt,yshift=-10pt]b-incoming boundary 1) {$\mu$};
	\end{tikzpicture} 
	}
	}
	\\
	&\eta: \mathbb{C} \rightarrow A &&\longleftrightarrow 
	&&\vcenter{
	\hbox{
	\begin{tikzpicture}[tqft/cobordism/.style={draw},
		tqft/view from=outgoing,
		tqft/boundary separation=20pt,
		tqft/cobordism height=40pt,
		tqft/circle x radius=7pt,
		tqft/circle y radius=4pt,
		tqft/every boundary component/.style={draw,rotate=0}]
   		\pic[tqft,
	   		incoming boundary components=0,
	   		outgoing boundary components=1,
			offset=.5,
	   		rotate=90,
			name=b,anchor={(0,0)}];
		\node at ([xshift=35pt,yshift=-8pt]b-incoming boundary 1) {$\eta$};
	\end{tikzpicture} 
	}
	}
	\\
	&\varepsilon: A \rightarrow \mathbb{C} &&\longleftrightarrow 
	&&\vcenter{
	\hbox{
	\begin{tikzpicture}[tqft/cobordism/.style={draw},
		tqft/view from=outgoing,
		tqft/boundary separation=20pt,
		tqft/cobordism height=40pt,
		tqft/circle x radius=7pt,
		tqft/circle y radius=4pt,
		tqft/every boundary component/.style={draw,rotate=0}]
   		\pic[tqft,
	   		incoming boundary components=1,
	   		outgoing boundary components=0,
			offset=.5,
	   		rotate=90,
			name=b,anchor={(0,0)}];
		\node at ([xshift=4pt,yshift=-14pt]b-incoming boundary 1) {$\varepsilon$};
	\end{tikzpicture} 
	}
	}
	\end{aligned}
~.~
\end{equation}
It is useful to use labelled boundaries to refer to specific components of the above maps.
For example, the structure constants $f_{ij}^k$ defined by
\begin{equation}
	e_i e_j = \sum_k f_{ij}^k e_k~,~
\end{equation}
correspond to the labelled pair of pants
\begin{equation}
	f_{ij}^k \longleftrightarrow 
	\vcenter{
	\hbox{
	\begin{tikzpicture}[tqft/cobordism/.style={draw},
		tqft/view from=outgoing,
		tqft/boundary separation=20pt,
		tqft/cobordism height=25pt,
		tqft/circle x radius=5pt,
		tqft/circle y radius=2pt,
		tqft/every boundary component/.style={draw,rotate=0}]
   		\pic[tqft,
	   		incoming boundary components=2,
	   		outgoing boundary components=1,
			offset=.5,
	   		rotate=90,
			name=b,anchor={(0,0)}];
   		\node at ([xshift=10pt]b-outgoing boundary 1) {$k$};
   		\node at ([xshift=-10pt]b-incoming boundary 1) {$i$};
   		\node at ([xshift=-10pt]b-incoming boundary 2) {$j$};
	\end{tikzpicture} 
	}
	}
.~
\end{equation}

Let $\tilde{g}^{ij}$ be the inverse of the matrix $g_{ij}$ defined in equation \eqref{eq: frob form}.
It plays a special role in 2D TQFTs, since it can be used to define a so-called handle creation element given by
\begin{equation}
	h = \sum_{i,j} \tilde{g}^{ij} e_i e_j~.~ 
	\label{eq: handle creation}
\end{equation}
In graphical notation, this corresponds to a surface with a single handle and one boundary as illustrated by
\begin{equation}
	\vcenter{
	\hbox{
	\begin{tikzpicture}[tqft/cobordism/.style={draw},
		tqft/view from=outgoing,
		tqft/boundary separation=20pt,
		tqft/cobordism height=30pt,
		tqft/circle x radius=6pt,
		tqft/circle y radius=3pt,
		tqft/every boundary component/.style={draw,rotate=0}]
   		\pic[tqft,
	   		incoming boundary components=2,
	   		outgoing boundary components=1,
			offset=.5,
	   		rotate=90,name=b];
   		\pic[tqft,
	   		incoming boundary components=0,
	   		outgoing boundary components=2,
	   		rotate=90,name=a,anchor=outgoing boundary 1,at=(b-incoming boundary 1)];
	\end{tikzpicture} 
	}
	}
~.~
\end{equation}
For this reason, $h$ can be used to create algebraic quantities that correspond to higher genus surfaces.
\\

Let us now focus on semi-simple commutative Frobenius algebras.
The Artin-Wedderburn theorem \cite{Wedderburn,Artin} for semi-simple algebras implies that there exists a basis of idempotents/projectors $\{P_a\}_{a=1}^K$ satisfying
\begin{equation}
	P_a P_b = \delta_{ab} P_a \,, \quad \sum_{a=1}^K P_a = 1\, .
\end{equation}
From \cite[Chapter 1 Theorem 3.8]{Ram1991}, we have that
\begin{equation}
	P_a = n_a \sum_{i,j} \widehat{\chi}^a(e_i) \tilde{g}^{ij} e_j~,~ \label{eq: fourier inversion}
\end{equation}
where $\widehat{\chi}^a$ is an irreducible character of $A$ and the normalisation constant $n_a$ is determined by $P_a P_a = P_a$.
Concrete examples of these normalisation constants are given in \eqref{eq: idem class} and \eqref{eq: idem fus}.
We can now prove the following proposition.
\begin{proposition}
	Let $A$ be a semi-simple commutative Frobenius algebra with bases $\{e_i\}_{i=1}^K$ and $\{P_a\}_{a=1}^K$ satisfying
	\begin{equation}
		\begin{aligned}
			e_i e_j &= \sum_k f_{ij}^k e_k	\, ,\quad
			P_a P_b &= \delta_{ab} P_a \, ,
		\end{aligned}
	\end{equation}
	then
	\begin{equation}
		e_i P_a = \widehat{\chi}^a(e_i) P_a \, , \label{eq: idempot EV}
	\end{equation}
	where $\widehat{\chi}^a(e_i)$ is an irreducible character of $A$ evaluted on $e_i$.
\end{proposition}
\begin{proof}
	From equation \eqref{eq: fourier inversion} we have
	\begin{equation}
		\begin{aligned}
			e_i P_a
			&= n_a \sum_{j,k} \widehat{\chi}^a(e_j) \tilde{g}^{jk} e_i e_k \\
			&= n_a \sum_{j,k,l} \widehat{\chi}^a(e_j) \tilde{g}^{jk} f_{ik}^l e_l \\
			&= n_a \sum_{j,k,l,m} \widehat{\chi}^a(e_j) \tilde{g}^{jk} g(e_i e_k, e_m) \tilde{g}^{ml} e_l \\
			&= n_a \sum_{j,k,l,m} \widehat{\chi}^a(e_j) \tilde{g}^{jk} g(e_i e_m, e_k) \tilde{g}^{ml} e_l \\
			&= n_a \sum_{l,m} \widehat{\chi}^a(e_i e_m) \tilde{g}^{ml} e_l \\
			&= \widehat{\chi}^a(e_i)  P_a \, .
		\end{aligned}
	\end{equation}
	where in the third equality we used the fact that $\tilde{g}$ is the inverse of $g$, the fourth equality uses $g(e_i e_j, e_k) = \varepsilon(e_i e_j e_k) = \varepsilon(e_i e_k e_j) = g(e_i e_k, e_j)$, and the last equality follows because $A$ is commutative and therefore $\widehat{\chi}^a(e_ie_m)=\widehat{\chi}^a(e_i) \widehat{\chi}^a(e_m)$ for irreducible characters. 
\end{proof}
As a consequence we have the inverse change of basis
\begin{equation}
	e_i = e_i 1 = e_i \sum_{a=1}^K P_a = \sum_{a=1}^K \widehat{\chi}^a(e_i) P_a \, .
\end{equation}

The eigenequation \eqref{eq: idempot EV} plays an important role in this paper.
In particular, we are interested in subsets of elements $t_1, \dots, t_l \in A$ that are sufficient to distinguish all idempotents/projectors $P_a$.
A generalisation of the arguments in \cite[Section 3.4]{KempRam} and \cite[Section 3.1]{SRES} gives the following theorem (see Appendix \ref{apx: generators} for the proof),
\begin{theorem}
\label{thm: generators}
Let $A$ be a commutative semi-simple algebra over $\mathbb{C}$, and $t_1, \dots, t_l$ be a sequence of elements in $A$.
The following two statements are equivalent:
\begin{enumerate}
	\item The lists of eigenvalues $\left(\widehat{\chi}^a(t_1), \dots, \widehat{\chi}^a(t_l)\right)$ are sufficent to distinguish every projector $P_a$ from all other projectors.
	In other words,
	\begin{equation}
		(\widehat{\chi}^a(t_1), \dots, \widehat{\chi}^a(t_l)) \neq 
		(\widehat{\chi}^b(t_1), \dots, \widehat{\chi}^b(t_l)) \Leftrightarrow P_a \neq P_b.
	\end{equation}
	\item The elements $t_1, \dots, t_l$ multiplicatively generate $A$. That is, every element $t \in A$ can be written as
	\begin{equation}
		t = \sum_{n_1, \dots, n_l \geq 0} C_{n_1 \dots n_l} t_1^{n_1} \dots t_l^{n_l}
	\end{equation}
	for a finite set of non-zero complex numbers $C_{n_1 \dots n_l}$. 
\end{enumerate}
\end{theorem}

In our applications, the sequence of elements $t_1, \dots, t_l$ are basis elements of $A$.
We are particularly interested in Frobenius algebras with special distinguished bases, which we refer to as combinatorial bases.
\begin{definition}[Combinatorial basis]
Let $A$ be a commutative Frobenius algebra. A basis $\mathcal{B} = \{e_1, \dots, e_K\}$ for $A$ is called combinatorial if 
\begin{enumerate}
	\item The unit is contained in $\mathcal{B}$, which we choose to be $e_1$, without loss of generality.
	\item The structure constants $f_{ij}^k$ are non-negative integers.
	\item The co-unit has the following property
	\begin{equation}
		\varepsilon(e_i) = q \delta_{1i} \, ,
	\end{equation}
	where $q$ is a positive rational number.
\end{enumerate}
\end{definition}
The combinatorial bases, which we study in Sections \ref{sec: DW} and \ref{sec: fusion}, have the following additional property.
Associated with each basis element there is an integer $d(e_i)$ satisfying
\begin{equation}
	d(e_i) d(e_j) = \sum_k f_{ij}^k d(e_k) \, .
\end{equation}
Therefore, the set of basis elements $e_k$ that can appear in the product
\begin{equation}
	e_i e_j = \sum_{k} f_{ij}^k e_k
\end{equation}
with non-zero coefficient $f_{ij}^k$ have $d(e_k) \leq d(e_i) d(e_j)$.
This can be interpreted as a type of locality in the space of dimensions, for small dimension objects.
The above definition has interesting similarities with the definition of weak fusion rings (see \cite[Definition 3.1.3, Example 3.1.9(iii) and Section 3.8]{EGNO}).

Given a commutative semi-simple Frobenius algebra $A$ and a combinatorial basis $\mathcal{B}$, we define a measure of complexity $N(A, \mathcal{B})$, associated with the pair $(A,\mathcal{B})$.
The definition of $N(A, \mathcal{B})$ is inspired by the complexity studied in \cite{JosephSRProj}.
\begin{definition}[Circle generator complexity]
\label{def: rank complexity}
Let $\mathcal{B}$ be a combinatorial basis for the commutative semi-simple Frobenius algebra $A$.
The circle generator complexity of $(A, \mathcal{B})$, denoted $N(A,\mathcal{B})$ (or simply $N(A)$), is the smallest positive integer such that there exists at least one subset $\{t_1, \dots, t_{N(A)}\} \subseteq \mathcal{B}$ with the property that 
\begin{equation}
	t_1, \dots, t_{N(A)}
\end{equation}
multiplicatively generate $A$.
Equivalently, 
\begin{equation}
	(\widehat{\chi}^a(t_1), \widehat{\chi}^a(t_2), \dots, \widehat{\chi}^a(t_{N(A)}))
\end{equation}
distinguishes all projectors $P_a$.
\end{definition}
Since the handle creation element plays a crucial role in 2D TQFTs,
we also define a supplemented generator complexity.
\begin{definition}[Circle-and-handle generator complexity]
\label{def: rank handle complexity}
Let $\mathcal{B}$ be a combinatorial basis for the commutative semi-simple Frobenius algebra $A$ with handle creation element $h$.
The circle-and-handle generator complexity of $(A, \mathcal{B})$, denoted $N^{\text{ch}}(A,\mathcal{B})$ (or simply $N^{\text{ch}}(A)$), is determined by the smallest positive integer $l$ such that there exists at least one subset $\{t_1, \dots, t_l\} \subseteq \mathcal{B}$ with the property that 
\begin{equation}
	h, t_1, \dots, t_l
\end{equation}
multiplicatively generate $A$. We define $N^{\text{ch}}(A) = l+1$.
\end{definition}

We now consider 2D TQFTs directly related to  the representation theory of finite groups.
In order to study them, it is useful to introduce some notation that is used throughout the paper.
Given that $G$ is a finite group, there is a corresponding finite set of conjugacy classes $C_1,\dots,C_K$. We use $\mathrm{Cl}(G)$ to denote the set of all conjugacy classes,
\begin{equation}
	\mathrm{Cl}(G) = \{C_1, C_2, \dots, C_K\} \, .
\end{equation}
A finite group $G$ has a finite number of isomorphism classes of irreducible representations $R_1, \dots, R_K$. We use $\mathrm{Irr}(G)$ to denote the complete set of non-isomorphic irreducible representations,
\begin{equation}
	\mathrm{Irr}(G) = \{R_1, R_2, \dots, R_K\} \, .
\end{equation}
For every irreducible representation $R \in \mathrm{Irr}(G)$, there is a corresponding character $\chi^R(g)$. We use $d_R = \chi^R(1)$ for the dimension of the representation $R$.
Since characters are class functions, it is useful to introduce the notation $\chi^R_C$ for the value of $\chi^R(g)$ for any $g \in C$.
\\

\subsection{Dijkgraaf-Witten TQFT}
\label{sec: DW}
Dijkgraaf-Witten (DW) theory \cite{Dijkgraaf:1989pz}, which is a gauge theory with a finite gauge group $G$, is the prototypical physical example of a 2D TQFT.
 DW theory satisfies the Atiyah-Segal axioms and the corresponding commutative Frobenius algebra is the so-called class algebra of $G$.

In order to understand the class algebra of a finite group $G$, 
it is useful to first review the group algebra of $G$.
The group algebra $\mathbb{C}[G]$ of a finite group $G$ is a complex vector space of dimension $|G|$ and has elements
\begin{equation}
	a = \sum_{g \in G} a_g g, \quad a_g \in \mathbb{C}
\end{equation}
with group multiplication determined by
\begin{equation}
	ab = \sum_{g,h \in G} a_g b_h gh \, .
\end{equation}
If $G$ is non-abelian, $\mathbb{C}[G]$ is a non-commutative algebra and is therefore not the algebra 
that is used in the Atiyah-Segal axioms.

However, the group algebra contains a commutative subalgebra called the center $\mathcal{Z}(\mathbb{C}[G])$ or the class algebra.
It has a combinatorial basis labelled by conjugacy classes $C \in \mathrm{Cl}(G)$.
Specifically, we define the class sums
\begin{equation}
	T_C = \sum_{g \in C} g \, ,
\end{equation}
which form a combinatorial basis for the class algebra given by
\begin{equation}
	\mathcal{Z}(\mathbb{C}[G]) = \mathrm{Span}(T_C : C \in \mathrm{Cl}(G)) \, .
\end{equation}
The structure constants in this basis, which are defined by
\begin{equation}
	T_{C} T_D = \sum_{E \in \mathrm{Cl}(G)} f_{CD}^E T_E \, ,
\end{equation}
can be written in terms of characters as follows
\begin{equation}
	f_{CD}^E = \frac{|C||D|}{|G|} \sum_{R \in \mathrm{Irr}(G)} \frac{\chi^R_C \chi^R_D \overline{\chi}^R_E }{d_R} \, .
\end{equation}
The conjugacy class of the identity $C_0={\textrm{id}}$ is the unit
\begin{equation}
	\eta(c) = c T_{C_0}
\end{equation}
and the co-unit is given by
\begin{equation}
	\varepsilon(T_C) = \frac{1}{|G|} \delta_{C C_0} \, .
\end{equation}
The class algebra is semi-simple and the idempotent/projector basis elements can be written explicitly in terms of characters of $G$.
For an irreducible representation $R \in \mathrm{Irr}(G)$, we define
\begin{equation}
	P_R = \frac{d_R}{|G|} \sum_{C \in \mathrm{Cl}(G)} \overline \chi^R_C T_C \, , \label{eq: idem class}
\end{equation}
where they satisfy
\begin{equation}
	P_R P_S = \delta_{RS} P_R \label{eq: CL projectors} \, ,
\end{equation}
which can be checked using orthogonality of characters.

Specialising equation \eqref{eq: idempot EV} to the class algebra gives 
\begin{equation}
	T_C P_R = \widehat{\chi}^R(T_C) P_R \, , \label{eq: class EV}
\end{equation}
for every $C \in \mathrm{Cl}(G)$ and $R \in \mathrm{Irr}(G)$.
The eigenvalues $\widehat{\chi}^R(T_C)$, also known as central or normalized characters, can be written in terms of characters and conjugacy class sizes of $G$
\begin{equation}
	\widehat{\chi}^R(T_C) = \frac{|C|\chi^R_C}{d_R} \, .
\end{equation}
The Atiyah-Segal axioms associate the normalized characters with the following labelled cobordisms
\begin{equation}
	\widehat{\chi}^R(T_C) = 
	\vcenter{
	\hbox{
	\begin{tikzpicture}[tqft/cobordism/.style={draw},
		tqft/view from=outgoing,
		tqft/boundary separation=20pt,
		tqft/cobordism height=25pt,
		tqft/circle x radius=5pt,
		tqft/circle y radius=2pt,
		tqft/every boundary component/.style={draw,rotate=0}]
   		\pic[tqft,
	   		incoming boundary components=2,
	   		outgoing boundary components=1,
			offset=.5,
	   		rotate=90,name=b,anchor={(0,0)}];
   		\node at ([xshift=10pt]b-outgoing boundary 1) {$R$};
   		\node at ([xshift=-10pt]b-incoming boundary 1) {$C$};
   		\node at ([xshift=-10pt]b-incoming boundary 2) {$R$};
	\end{tikzpicture} 
	}
	}
	\, .
\end{equation}

In order to compute the handle creation element $H$ in a DW theory, we use the definition in \eqref{eq: handle creation}.
In the projector basis, the matrix of the bilinear form is diagonal and it is straightforward to compute the inverse (see \cite[Section 2.1]{RowCol}). We have 
\begin{equation}
	g(P_R, P_S) = \delta_{RS} \frac{d_R^2}{|G|^2} \quad \Rightarrow \quad \tilde{g}^{RS} = \delta^{RS} \frac{|G|^2}{d_R^2}\, .
\end{equation}
Therefore, the handle creation element is given by 
\begin{equation}
	H = \sum_{R,S} \tilde{g}^{RS} P_R P_S = \sum_{R \in \mathrm{Irr}(G)} \frac{|G|^2}{d_R^2} P_R \, . \label{eq: handle creation DW}
\end{equation}
Equation \eqref{eq: handle creation DW} can be understood geometrically as
\begin{equation}
	\frac{|G|^2}{d_R^2}  =
	\vcenter{
	\hbox{
	\begin{tikzpicture}[tqft/cobordism/.style={draw},
		tqft/view from=outgoing,
		tqft/boundary separation=20pt,
		tqft/cobordism height=30pt,
		tqft/circle x radius=6pt,
		tqft/circle y radius=3pt,
		tqft/every boundary component/.style={draw,rotate=0}]
   		\pic[tqft,
	   		incoming boundary components=2,
	   		outgoing boundary components=1,
			offset=.5,
	   		rotate=90,name=b];
   		\pic[tqft,
	   		incoming boundary components=0,
	   		outgoing boundary components=2,
	   		rotate=90,name=a,anchor=outgoing boundary 1,at=(b-incoming boundary 1)];
   		\node at ([xshift=10pt]b-outgoing boundary 1) {$R$};
	\end{tikzpicture} 
	}
	}
\, .
\end{equation}
From the projector properties \eqref{eq: CL projectors}, it follows that 
\begin{equation}
	H P_R = \frac{|G|^2}{d_R^2} P_R \, . \label{eq: DW handle EV}
\end{equation}

We summarise the most important data of a DW theory as follows:
\begin{equation*}
	\begin{aligned}
	&\mu: A \otimes A \rightarrow A &&\longleftrightarrow 
	&&\vcenter{
	\hbox{
	\vspace{-1em}
	\begin{tikzpicture}[tqft/cobordism/.style={draw},
		tqft/view from=outgoing,
		tqft/boundary separation=20pt,
		tqft/cobordism height=25pt,
		tqft/circle x radius=5pt,
		tqft/circle y radius=2pt,
		tqft/every boundary component/.style={draw,rotate=0}]
   		\pic[tqft,
	   		incoming boundary components=2,
	   		outgoing boundary components=1,
			offset=.5,
	   		rotate=90,
			name=b,
			anchor={(0,0)}];
		\node at ([xshift=12pt,yshift=-10pt]b-incoming boundary 1) {$\mu$};
	\end{tikzpicture} 
	}
	}
	\\
	&\mu(T_C,T_D) = \sum_{E\in \mathrm{Cl}(G)} f_{CD}^E T_E
	\\[2em]
	&\eta: \mathbb{C} \rightarrow A &&\longleftrightarrow 
	&&\vcenter{
	\hbox{
	\vspace{-1em}
	\begin{tikzpicture}[tqft/cobordism/.style={draw},
		tqft/view from=outgoing,
		tqft/boundary separation=20pt,
		tqft/cobordism height=40pt,
		tqft/circle x radius=7pt,
		tqft/circle y radius=4pt,
		tqft/every boundary component/.style={draw,rotate=0}]
   		\pic[tqft,
	   		incoming boundary components=0,
	   		outgoing boundary components=1,
			offset=.5,
	   		rotate=90,
			name=b,anchor={(0,0)}];
		\node at ([xshift=35pt,yshift=-8pt]b-incoming boundary 1) {$\eta$};
	\end{tikzpicture} 
	}
	}
	\\
	&\eta(1) = T_{C_0}
	\\[2em]
	&\varepsilon: A \rightarrow \mathbb{C} &&\longleftrightarrow 
	&&\vcenter{
	\hbox{
	\vspace{-1em}
	\begin{tikzpicture}[tqft/cobordism/.style={draw},
		tqft/view from=outgoing,
		tqft/boundary separation=20pt,
		tqft/cobordism height=40pt,
		tqft/circle x radius=7pt,
		tqft/circle y radius=4pt,
		tqft/every boundary component/.style={draw,rotate=0}]
   		\pic[tqft,
	   		incoming boundary components=1,
	   		outgoing boundary components=0,
			offset=.5,
	   		rotate=90,
			name=b,anchor={(0,0)}];
		\node at ([xshift=4pt,yshift=-14pt]b-incoming boundary 1) {$\varepsilon$};
	\end{tikzpicture} 
	}
	}
	\\
	&\varepsilon(T_C) = \frac{\delta_{C {C_0}}}{|G|}
	\\[2em]
	&H = \sum_{R \in \mathrm{Irr}(G)} \frac{|G|^2}{d_R^2} P_R &&\longleftrightarrow 
	&&\vcenter{
	\hbox{
	\begin{tikzpicture}[tqft/cobordism/.style={draw},
		tqft/view from=outgoing,
		tqft/boundary separation=20pt,
		tqft/cobordism height=30pt,
		tqft/circle x radius=6pt,
		tqft/circle y radius=3pt,
		tqft/every boundary component/.style={draw,rotate=0}]
   		\pic[tqft,
	   		incoming boundary components=2,
	   		outgoing boundary components=1,
			offset=.5,
	   		rotate=90,name=b];
   		\pic[tqft,
	   		incoming boundary components=0,
	   		outgoing boundary components=2,
	   		rotate=90,name=a,anchor=outgoing boundary 1,at=(b-incoming boundary 1)];
	\end{tikzpicture} 
	}
	}
	\end{aligned}
\end{equation*}
\\

\subsection{Measures of complexity in DW TQFTs}
\label{sec: class complexity}
In \cite{JosephSRProj}, the complexity of quantum projector detection algorithms was studied for projectors in $A = \mathcal{Z}(\mathbb{C}[S_n])$.
The detection algorithm is based on a quantum phase estimation subroutine (see \cite[Section 5.2]{nielsen2010quantum}).
In particular, given a generating set $\{T_{C_i}\}_{i=1}^{l}$ of $\mathcal{Z}(\mathbb{C}[S_n])$ it is necessary to apply the subroutine for each element of the generating set, leading to a $l$-dependence of the complexity.

In a parallel development \cite{PRS2024}, generalisations of eigenequations like \eqref{eq: class EV} have proven to be important in  classical algorithms for the construction of bases for multi-matrix invariants, which have applications in AdS/CFT.
There, the conjugacy class sizes $|C|$ associated with the combinatorial basis elements $T_C$ played a role in the complexity of the classical algorithms.
We now introduce these measures of complexity for general groups.

Since the class algebra has combinatorial basis elements $T_C$ and the eigenvalues in \eqref{eq: class EV} are given by central characters, applying Theorem \ref{thm: generators} and \eqref{eq: class EV} gives the following.
\begin{proposition}
	\label{prop: class algebra generating set}
	Let $\{D_1, \dots, D_l\} \subseteq \mathrm{Cl}(G)$ and let $\{T_{D_1}, \dots, T_{D_l}\}$ be the corresponding combinatorial basis elements.
	They multiplicatively generate the class algebra $\mathcal{Z}(\mathbb{C}[G])$ if the list
	\begin{equation}
		\left(\frac{|D_1| \chi^R_{D_1}}{d_R}, \dots, \frac{|D_l| \chi^R_{D_l} }{d_R}\right)
	\end{equation}
	distinguishes all irreducible representations $R \in \mathrm{Irr}(G)$.
\end{proposition}
It is also useful to have the analogous result for a generating sets of combinatorial basis elements alongside the handle creation element $H$.
\begin{proposition}
	\label{prop: class algebra generating set handle}
	Let $\{D_1, \dots, D_l\} \subseteq \mathrm{Cl}(G)$, $\{T_{D_1}, \dots, T_{D_l}\}$ be the corresponding combinatorial basis elements, and $H$ be the handle creation operator \eqref{eq: handle creation DW}.
	The set $\{H, T_{D_1}, \dots, T_{D_l}\}$ multiplicatively generates the class algebra $\mathcal{Z}(\mathbb{C}[G])$ if the list
	\begin{equation}
    	\left(\frac{|G|^2}{d_R^2}, \frac{|D_1| \chi^R_{D_1}}{d_R}, \dots, \frac{|D_l| \chi^R_{D_l} }{d_R}\right)
	\end{equation}
	distinguishes all irreducible representations $R \in \mathrm{Irr}(G)$.
\end{proposition}
\begin{proof}
	This follows from Theorem \ref{thm: generators} and the eigenvalue equation in \eqref{eq: DW handle EV}.
\end{proof}

Studying the circle-and-handle complexity is directly related to studying the character table of $G$.
In particular, the list
\begin{equation}
	\left(\frac{|G|^2}{d_R^2}, \frac{|D_1| \chi^R_{D_1}}{d_R}, \dots, \frac{|D_l| \chi^R_{D_l} }{d_R}\right)
\end{equation}
distinguishes all projectors $P_R$ if and only if the list
\begin{equation}
	(d_R, \chi^R_{D_1}, \dots, \chi^R_{D_l}) \label{eq: circle handle list}
\end{equation}
distinguishes all projectors $P_R$.
To see this, let $R,S \in \mathrm{Irr}(G)$ be distinguishable by the first list but not the second.
That is, 
\begin{equation}
	\frac{|G|^2}{d_R^2} \neq \frac{|G|^2}{d_S^2} \quad \text{ or \quad $\exists j \in \{1,\dots,l\}$ such that } \quad \frac{|D_j|\chi^R_{D_j}}{d_R} \neq \frac{|D_j|\chi^S_{D_j}}{d_S}
\end{equation}
but
\begin{equation}
	d_R = d_S, \quad \chi^R_{D_1} = \chi^S_{D_1} \quad ,\dots, \quad \chi^R_{D_l} = \chi^S_{D_l} \, .
\end{equation}
This is a contradiction since it implies
\begin{equation}
		\frac{|G|^2}{d_R^2} = \frac{|G|^2}{d_S^2} \, ,\quad \frac{|D_j|\chi^R_{D_j}}{d_R} = \frac{|D_j|\chi^S_{D_j}}{d_S} \, .
\end{equation}
The other direction follows from a similar argument by contradiction.

For the sake of readability, we introduce the following short-hand notation for the generator complexities in Definition \ref{def: rank complexity} and \ref{def: rank handle complexity} specialised to $A=\mathcal{Z}(\mathbb{C}[G])$,
\newcommand*{\brankC}{N_{\text{cls}}}
\newcommand*{\bhrankC}{N^{\text{ch}}_{\text{cls}}}
\begin{equation}
	\begin{aligned}
		\brankC(G)&\equiv N(\mathcal{Z}(\mathbb{C}[G])) 	\, ,\\
		\bhrankC(G) &\equiv N^{\text{ch}}(\mathcal{Z}(\mathbb{C}[G])) 	\, .
	\end{aligned}
\label{eq: complexity boundary class}
\end{equation}
These generator complexities correspond to the minimal values of $l$ for which we can apply Proposition \ref{prop: class algebra generating set} and \ref{prop: class algebra generating set handle}.

We also study the average class size over the minimal generating subsets, which is defined as follows.
\begin{definition}[Average class size over  minimal generating subsets]
Consider a class algebra $\mathcal{Z}(\mathbb{C}[G])$ with combinatorial basis $\mathcal{B} = \{T_C \, : \, C \in \mathrm{Cl}(G)\}$ and circle generator complexity $\brankC(G)$.
In general, there can be several distinct choices of minimal generating subsets $F \subseteq \mathcal{B}$ with $|F| = \brankC(G)$.
Let $k$ be the total number of minimal generating subsets and
\begin{equation}
	F_1, \dots, F_k \subseteq \mathcal{B} 
\end{equation}
be the corresponding subsets.
We define the average class size over the minimal  generating subsets as
\begin{equation}
	\mathcal{C}_{\text{gens}}(G) = \frac{1}{k\brankC(G)} \sum_{i=1}^k \sum_{T_C \in F_i} |C| \, . \label{eq: avg gens class size}
\end{equation}
\end{definition}

\subsection{Fusion TQFT}
\label{sec: fusion}
The fusion TQFT is based on the fusion algebra of a finite group.
This TQFT was understood to be row-column dual to the DW theory in \cite{RowCol}.
We briefly review the construction of fusion TQFTs.

Let $R,S,T \in \mathrm{Irr}(G)$, and
\begin{equation}
	N_{RS}^T = \frac{1}{|G|} \sum_{g \in G} \chi^R(g) \chi^S(g) \chi^T(g^{-1}) \, ,
\end{equation}
be the multiplicity of the irreducible representation $T$ in the decomposition of the tensor product $R \otimes S$.
The fusion algebra has a combinatorial basis given by
\begin{equation}
	R(G) = \mathrm{Span}( \, a_R ~ :~ R \in \mathrm{Irr}(G) \, )
\end{equation}
with structure constants
\begin{equation}
	a_R a_S = \sum_{T \in \mathrm{Irr}(G)} N^T_{RS} a_T \, . 
\end{equation}
The unit corresponds to the trivial representation $R_0$ of $G$,
\begin{equation}
	\eta(c) = c a_{R_0}
\end{equation}
and the co-unit is given by
\begin{equation}
	\varepsilon(a_R) = \delta_{R R_0}\, .
\end{equation}
This commutative algebra also has a basis of projectors.
For every $C \in \mathrm{Cl}(G)$, we define
\begin{equation}
	A_C = \frac{|C|}{|G|} \sum_{R \in \mathrm{Irr}(G)} \overline{\chi}^R_C a_R\, . \label{eq: idem fus}
\end{equation}
They satisfy (see \cite{RowCol} for a proof)
\begin{equation}
	A_C A_D = \delta_{CD} A_C \quad \forall C,D \in \mathrm{Cl}(G) \, .\label{eq: fusion idem}
\end{equation}
The two bases are connected by a set of eigenequations
\begin{equation}
	a_R A_C = \chi^R_C A_C \, . \label{eq: fusion EV}
\end{equation}
The handle creation element, which we denote by $H$, is computed by using Definition \eqref{eq: handle creation} in the projector basis.
From \cite[Section 2.3]{RowCol} we have
\begin{equation}
	g(A_C, A_D) = \delta_{CD} \frac{|C|}{|G|} \quad \Rightarrow \quad \tilde{g}^{CD} = \delta^{CD} \frac{|G|}{|C|}\, .
\end{equation}
Therefore the handle creation element in the fusion algebra takes the form,
\begin{equation}
	H = \sum_{C \in \mathrm{Cl}(G)} \frac{|G|}{|C|} A_C\, . \label{eq: handle creation fusion}
\end{equation}
Using equation \eqref{eq: fusion idem} gives the eigenequation
\begin{equation}
	H A_C = \frac{|G|}{|C|} A_C\, . \label{eq: fusion handle EV}
\end{equation}

We summarise the most important data of of the fusion TQFT as follows:
\begin{equation*}
	\begin{aligned}
	&\mu: A \otimes A \rightarrow A &&\longleftrightarrow 
	&&\vcenter{
	\hbox{
	\vspace{-1em}
	\begin{tikzpicture}[tqft/cobordism/.style={draw},
		tqft/view from=outgoing,
		tqft/boundary separation=20pt,
		tqft/cobordism height=25pt,
		tqft/circle x radius=5pt,
		tqft/circle y radius=2pt,
		tqft/every boundary component/.style={draw,rotate=0}]
   		\pic[tqft,
	   		incoming boundary components=2,
	   		outgoing boundary components=1,
			offset=.5,
	   		rotate=90,
			name=b,
			anchor={(0,0)}];
		\node at ([xshift=12pt,yshift=-10pt]b-incoming boundary 1) {$\mu$};
	\end{tikzpicture} 
	}
	}
	\\
	&\mu(a_R,a_S) = \sum_{T \in \mathrm{Irr}(G)} N_{RS}^T a_T
	\\[2em]
	&\eta: \mathbb{C} \rightarrow A &&\longleftrightarrow 
	&&\vcenter{
	\hbox{
	\vspace{-1em}
	\begin{tikzpicture}[tqft/cobordism/.style={draw},
		tqft/view from=outgoing,
		tqft/boundary separation=20pt,
		tqft/cobordism height=40pt,
		tqft/circle x radius=7pt,
		tqft/circle y radius=4pt,
		tqft/every boundary component/.style={draw,rotate=0}]
   		\pic[tqft,
	   		incoming boundary components=0,
	   		outgoing boundary components=1,
			offset=.5,
	   		rotate=90,
			name=b,anchor={(0,0)}];
		\node at ([xshift=35pt,yshift=-8pt]b-incoming boundary 1) {$\eta$};
	\end{tikzpicture} 
	}
	}
	\\
	&\eta(1) = a_{R_0}
	\\[2em]
	&\varepsilon: A \rightarrow \mathbb{C} &&\longleftrightarrow 
	&&\vcenter{
	\hbox{
	\vspace{-1em}
	\begin{tikzpicture}[tqft/cobordism/.style={draw},
		tqft/view from=outgoing,
		tqft/boundary separation=20pt,
		tqft/cobordism height=40pt,
		tqft/circle x radius=7pt,
		tqft/circle y radius=4pt,
		tqft/every boundary component/.style={draw,rotate=0}]
   		\pic[tqft,
	   		incoming boundary components=1,
	   		outgoing boundary components=0,
			offset=.5,
	   		rotate=90,
			name=b,anchor={(0,0)}];
		\node at ([xshift=4pt,yshift=-14pt]b-incoming boundary 1) {$\varepsilon$};
	\end{tikzpicture} 
	}
	}
	\\
	&\varepsilon(a_R) = \delta_{R {R_0}}
	\\[2em]
	&H = \sum_{C \in \mathrm{Cl}(G)} \frac{|G|}{|C|} A_C &&\longleftrightarrow 
	&&\vcenter{
	\hbox{
	\begin{tikzpicture}[tqft/cobordism/.style={draw},
		tqft/view from=outgoing,
		tqft/boundary separation=20pt,
		tqft/cobordism height=30pt,
		tqft/circle x radius=6pt,
		tqft/circle y radius=3pt,
		tqft/every boundary component/.style={draw,rotate=0}]
   		\pic[tqft,
	   		incoming boundary components=2,
	   		outgoing boundary components=1,
			offset=.5,
	   		rotate=90,name=b];
   		\pic[tqft,
	   		incoming boundary components=0,
	   		outgoing boundary components=2,
	   		rotate=90,name=a,anchor=outgoing boundary 1,at=(b-incoming boundary 1)];
	\end{tikzpicture} 
	}
	}
	\end{aligned}
\end{equation*}
\\

\subsection{Measures of complexity in fusion TQFTs}
\label{sec: fusion complexity}
We now specialise the generator complexities in Definitions \ref{def: rank complexity} and \ref{def: rank handle complexity} to fusion algebras. This gives measures of complexity that are row-column dual to the complexities associated with the class algebra. 

First, we apply Theorem \ref{thm: generators} connecting generators and eigenvalues to the fusion algebra.
Since the combinatorial basis of the fusion algebra is given by elements $a_R$ and the eigenvalues in \eqref{eq: fusion EV} are characters $\chi^R_C$, we get the following.
\begin{proposition}
	\label{prop: fusion algebra generating set}
	Let $\{S_1, \dots, S_l\} \subseteq \mathrm{Irr}(G)$ and $\{a_{S_1}, \dots, a_{S_l}\}$ be the corresponding combinatorial basis elements.
	They multiplicatively generate the fusion algebra $R(G)$ if the list
	\begin{equation}
		(\chi^{S_1}_C, \dots, \chi^{S_l}_C)
	\end{equation}
	distinguishes all conjugacy classes $C \in \mathrm{Cl}(G)$.
\end{proposition}
We also have the analogous result for a generating set including combinatorial basis elements alongside the handle creation element $H$.
\begin{proposition}
	\label{prop: fusion algebra generating set handles}
	Let $\{S_1, \dots, S_l\} \subseteq \mathrm{Irr}(G)$ , $\{a_{S_1}, \dots, a_{S_l}\}$ be the corresponding combinatorial basis elements, and $H$ be the handle creation operator \eqref{eq: handle creation fusion}.
	The set $\{H, a_{S_1}, \dots, T_{S_l}\}$ multiplicatively generates the fusion algebra $R(G)$ if the list
	\begin{equation}
		\left(\frac{|G|}{|C|}, \chi^{S_1}_C, \dots, \chi^{S_l}_C\right)
	\end{equation}
	distinguishes all conjugacy classes $C \in \mathrm{Cl}(G)$.
\end{proposition}
\begin{proof}
	This follows from Theorem \ref{thm: generators} and the eigenvalue equation in \eqref{eq: fusion handle EV}.
\end{proof}

Studying the circle-and-handle complexity corresponds to studying the character table supplemented by conjugacy class sizes.
In particular,
the list
\begin{equation}
	\left(\frac{|G|}{|C|}, \chi^{S_1}_C, \dots, \chi^{S_l}_C\right)
\end{equation}
distinguishes all projectors $A_C$ if and only if the list
\begin{equation}
	(|C|, \chi^{S_1}_C, \dots, \chi^{S_l}_C)\label{eq: circle handle list 2}
\end{equation}
distinguishes all projectors $A_C$.
This is the row-column dual version of equation \eqref{eq: circle handle list}.

Similarly to equation \eqref{eq: complexity boundary class}, we introduce the following notation
\newcommand*{\brankF}{N_{\text{fus}}}
\newcommand*{\bhrankF}{N^{\text{ch}}_{\text{fus}}}
\begin{equation}
	\begin{aligned}
		\brankF(G) &\equiv N(R(G)) \, , \\
		\bhrankF(G) &\equiv N^{\text{ch}}(R(G)) \, ,
	\end{aligned}
\label{eq: complexity boundary fusion}
\end{equation}
where the right-hand sides are the complexities defined in Definition \ref{def: rank complexity} and \ref{def: rank handle complexity} for $A=R(G)$.
These generator complexities correspond to the minimal values of $l$ for which we can apply Proposition \ref{prop: fusion algebra generating set} and \ref{prop: fusion algebra generating set handles} .

The following definition introduces the average squared dimension associated with generators of $R(G)$.
\begin{definition}[Average squared dimension over minimal generating subsets]
Consider a fusion algebra $R(G)$ with combinatorial basis $\mathcal{B} = \{ a_R \, : \, R \in \mathrm{Irr}(G)\}$ and circle generator complexity $\brankF(G)$.
In general, there can be several distinct choices of minimal generating subsets $F \subseteq \mathcal{B}$ with $|F| = \brankF(G)$.
Let $k$ be the total number of minimal generating subsets and
\begin{equation}
	F_1, \dots, F_k \subseteq \mathcal{B}
\end{equation}
be the corresponding subsets.
We define the average squared dimension over the minimal generating subsets as
\begin{equation}
	\mathcal{R}_{\text{gens}}(G) = \frac{1}{k\brankF(G)} \sum_{i=1}^k \sum_{a_R \in F_i} d_R^2 \, .\label{eq: avg gens dim squared}
\end{equation}
\end{definition}

\section{Statistics}
\label{sec: stats}
Equations \eqref{eq: circle handle list} and \eqref{eq: circle handle list 2} show that the generator complexity of class algebras and fusion algebras are directly related to properties of character tables.
In particular, the generator complexity of class algebras involves subsets of columns of the character table which distinguish all the rows while the generator complexity of fusion algebras involves subsets of rows which distinguish all the columns. 
Below, we have used a subset of character tables in the GAP character table library \cite{CTblLib} to study the measures of complexity defined in Sections \ref{sec: class complexity} and \ref{sec: fusion complexity}.
Sage \cite{sagemath} and GAP \cite{GAP4} code for generating the data used in this section is attached alongside the arXiv version of this paper.

In general, distinct groups can have the same character table and therefore the same class algebra and fusion algebra.
However, in the following sections it is useful to think of $G$ as a label for  character tables.
All uses of this label $G$ are independent of the choice of a representative group.
The dataset extracted from \cite{CTblLib} only uses a single representative per character table, there are no duplicate or equivalent character tables.
The following definition is used repeatedly in the following subsections.
\begin{definition}[Subset average]
	Let $S$ be a subset of all character tables and $f(G)$ be a function on character tables.
	We define the subset average of $f$ as
	\begin{equation}
		\mathbb{E}(S,f) = \frac{1}{|S|} \sum_{G \in S} f(G) \, .
	\end{equation}
\end{definition}

\subsection{Generator complexity}
The first statistic that we study is the average generator complexity as a function of $|\mathrm{Cl}(G)| = |\mathrm{Irr}(G)|$.
Let $K_n$ be the subset of all inequivalent character tables of size $n=|\mathrm{Cl}(G)|=|\mathrm{Irr}(G)|$.
We first study the following average circle generator complexities, defined in \eqref{eq: complexity boundary class} and \eqref{eq: complexity boundary fusion}
\begin{equation}
	\begin{aligned}
		&\mathbb{E}({K_n}, \brankC) = \frac{1}{|K_n|} \sum_{G \in K_n} \brankC{(G)} \, ,\\ 
		&\mathbb{E}(K_n, \brankF) = \frac{1}{|K_n|} \sum_{G \in K_n} \brankF{(G)} \, .
\end{aligned}
\end{equation}
In Figure \ref{fig: no gens boundary}, we plot these averages for the subset of character tables in \cite{CTblLib} with $n=|\mathrm{Cl}(G)|=|\mathrm{Irr}(G)| \in \{2,\dots, 30\}$.

\begin{figure}
	\centering
	\includegraphics[width=\textwidth]{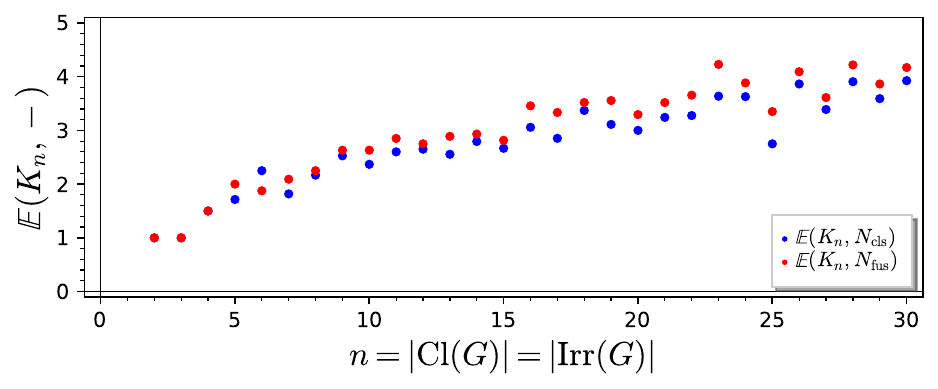}
	\caption{
			In blue, we have the average circle generator complexity of class algebras $\mathbb{E}(K_n, \brankC)$ of dimension $n = |\mathrm{Cl}(G)| \in\{2,\dots,30\}$.
			In red, we have the average circle generator complexity of fusion algebras $\mathbb{E}(K_n, \brankF)$ of dimension $n = |\mathrm{Irr}(G)| \in \{2, \dots, 30\}$. 
			}
	\label{fig: no gens boundary}
\end{figure}

By inspecting Figure \ref{fig: no gens boundary}, we see that the class algebra tends to have smaller average circle generator complexity at fixed $n$, with exceptions at $n=2,3,4,6$.
In particular, we have the following proposition.
\begin{proposition}
	Let $K_n$ be the set of inequivalent character tables of size $n$.
	The average circle generator complexity of class algebras tends to be smaller than the average circle generator complexity of fusion algebras.
	Specifically, we have
	\begin{equation}
		\mathbb{E}(K_n, \brankC) < \mathbb{E}(K_n, \brankF), \quad \text{for $n=5,7,8,\dots,30$}
	\end{equation}
	and 
	\begin{equation}
		\begin{aligned}
			&\mathbb{E}(K_n, \brankC) > \mathbb{E}(K_n, \brankF) , \quad \text{for $n=6$} \\
			&\mathbb{E}(K_n, \brankC) = \mathbb{E}(K_n, \brankF) , \quad \text{for $n=2,3,4$}\, .
		\end{aligned}
	\end{equation}
	\label{prop: boundary complexity}
\end{proposition}

Secondly, we study the average circle-and-handle generator complexities (defined in \eqref{eq: complexity boundary class}, \eqref{eq: complexity boundary fusion}),
\begin{equation}
	\begin{aligned}
		&\mathbb{E}(K_n, \bhrankC)  = \frac{1}{|K_n|} \sum_{G \in K_n} \bhrankC{(G)} \, , \\
		&\mathbb{E}(K_n, \bhrankF) = \frac{1}{|K_n|} \sum_{G \in K_n} \bhrankF{(G)} \, .
\end{aligned}
\end{equation}
We plot these averages in Figure \ref{fig: no gens bulk} for the subset of character tables in \cite{CTblLib} with $n=|\mathrm{Cl}(G)|=|\mathrm{Irr}(G)| \in \{2,\dots, 30\}$.
\begin{figure}
	\centering
	\includegraphics[width=\textwidth]{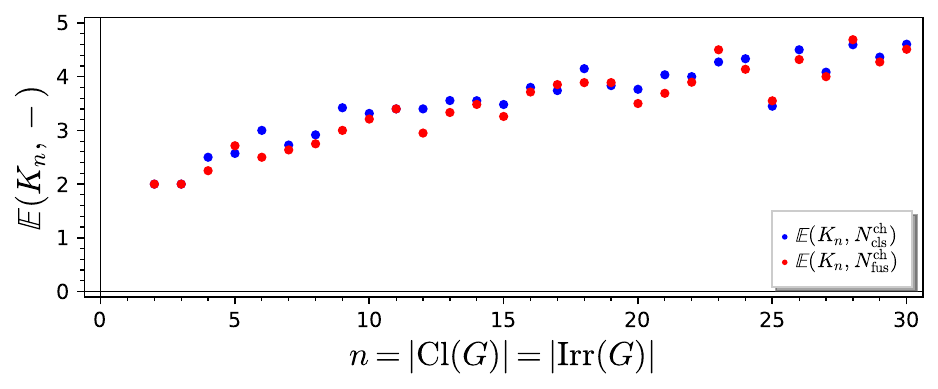}
	\caption{
			In blue, we have the average circle-and-handle generator complexity of class algebras $\mathbb{E}(K_n, \bhrankC)$ of dimension $n = |\mathrm{Cl}(G)| \in\{2,\dots,30\}$.
			In red, we have the average circle-and-handle generator complexity of fusion algebras $\mathbb{E}(K_n, \bhrankF)$ of dimension $n = |\mathrm{Irr}(G)| \in \{2, \dots, 30\}$.
	}
	\label{fig: no gens bulk}
\end{figure}
For this average, the relationship is flipped, but there are also more exceptions.
Inspecting the figure shows the following proposition.
\begin{proposition}
	Let $K_n$ be the set of inequivalent character tables of size $n$.
	For the majority of points the average circle-and-handle generator complexity of class algebras is larger than the corresponding average for fusion algebras.
	That is,
	\begin{equation}
		\begin{aligned}
			&\mathbb{E}(K_n, \bhrankC) > \mathbb{E}(K_n, \bhrankF)\, , \\
			&\text{for $n\in \{2,\dots,30\}\backslash\{2,3,5,11,17,19,23,25,28\}$}
		\end{aligned}
	\end{equation}
	and
	\begin{equation}
		\begin{aligned}
			&\mathbb{E}(K_n, \bhrankC) < \mathbb{E}(K_n, \bhrankF) , \quad \text{for $n=5,17,19,23,25,28$} \\
			&\mathbb{E}(K_n, \bhrankC) = \mathbb{E}(K_n, \bhrankF) , \quad \text{for $n=2,3,11$} \, .
		\end{aligned}
	\end{equation}
\end{proposition}

In Proposition \ref{prop: boundary complexity}, we see that the average circle generator complexity of class algebras is smaller than the corresponding quantity for fusion algebras for all $n>6$ in Figure \ref{fig: no gens boundary}.
Based on this observation, we make the following conjectures.
\begin{conjecture}[Circle generator complexity (Strong version)]
	The strong version says that the average circle generator complexity of class algebras is smaller than the corresponding average for fusion algebras, 
	\label{conj: gen strong}
	\begin{equation}
		\mathbb{E}(K_n, \brankC) < \mathbb{E}(K_n, \brankF) 
	\end{equation}
	for all $n>6$.
\end{conjecture}
A weaker version in the same spirit conjectures the existences of a finite transition point different from $n=6$.
\begin{conjecture}[Circle generator complexity (Weak version)]
	\label{conj: gen weak}
	There exists a finite $m \in \mathbb{N}$ such that for every $n > m$ the average circle generator complexity of class algebras is smaller than the average circle generator complexity of fusion algebras.
	That is,
	\begin{equation}
		\mathbb{E}(K_n, \brankC) < \mathbb{E}(K_n, \brankF), \quad \forall n > m \,.
	\end{equation}
\end{conjecture}

For the circle-and-handle generators, exceptions continue to occur in Figure \ref{fig: no gens bulk} for all $n$.
Therefore, we do not expect similar conjectures to hold in this case, but this is an interesting open question to study in the future. 
\begin{open}[Circle-and-handle generator complexity]
	Do exceptions to
	\begin{equation}
		\mathbb{E}(K_n, \bhrankC) > \mathbb{E}(K_n, \bhrankF) 
	\end{equation}
	continue to exist for larger $n$?
\end{open}

Based on the data presented in this subsection, the average circle generator complexity of class algebras tends to be smaller than the average circle generator complexity of fusion algebras.
A similar statement does not seem to be true for the average circle-and-handle generator complexity.
\\

\subsection{Average class size and dimension}
The following section studies the connection between conjugacy class sizes and generators of the class algebra.
The dual question involves studying dimensions of irreducible representations and generators of the fusion algebra.

First, it is useful to introduce a normalizing constant.
The average class size of a group $G$ is 
\begin{equation}
	\mathcal{N}(G) = \frac{1}{|\mathrm{Cl}(G)|} \sum_{C \in \mathrm{Cl}(G)} |C| = \frac{|G|}{|\mathrm{Cl}(G)|} \, ,
\end{equation}
which is equal to the average squared dimension given by
\begin{equation}
	\frac{1}{|\mathrm{Irr}(G)|} \sum_{R \in Irr(G)} d_R^2 = \frac{|G|}{|Irr(G)|} = \mathcal{N}(G) \, .
\end{equation}
This is an important group invariant in the following work.

Let us introduce here the averages that are relevant for the following work.
Let $X_N$ be the set of all inequivalent character tables of finite groups up to size $N$,
and $B_{N,n} \subseteq X_N$ be the subset of character tables with average class size $\mathcal{N}(G) = n$ given by
\begin{equation}
	B_{N,n} = \{ G \in X_N \quad \text{s.t.} \quad \mathcal{N}(G) = n\} \, .
\end{equation}
We study the following averages of the functions defined in \eqref{eq: avg gens class size} and \eqref{eq: avg gens dim squared}, as functions of $n$
\begin{equation}
	\begin{aligned}
		\mathbb{E}(B_{N,n},\mathcal{C}_{\text{gens}}) &= \frac{1}{|B_{N,n}|} \sum_{G \in B_{N,n}} \mathcal{C}_{\text{gens}}(G) \, , \\
		\mathbb{E}(B_{N,n}, \mathcal{R}_{\text{gens}}) &= \frac{1}{|B_{N,n}|} \sum_{G \in B_{N,n}} \mathcal{R}_{\text{gens}}(G) \, .
	\end{aligned}
\end{equation}
In Figures \ref{fig: avg gen class size} and \ref{fig: avg gen dim}, we plot these averages for the set of character tables $X_{30}$, taken from \cite{CTblLib}.
We observe that there are roughly as many points above and below the lines $y(n)\equiv\mathbb{E}(B_{30,n},\mathcal{C}_{\text{gens}})=n$ and $y(n)\equiv\mathbb{E}(B_{30,n},\mathcal{R}_{\text{gens}})=n$ over sizable ranges of $n$, respectively. 
To quantify this, we study the number of accumulated points above and below the line $y(n)=n$, up to fixed values of $n$, in Figures \ref{fig: avg gen class size} and \ref{fig: avg gen dim}.
We use the following definition to formulate our observation here.
\begin{definition}
	Let $l$ be a fixed real number. We define the cumulative ratio of the number of points above and below the line $y(n)=n$ in Figure \ref{fig: avg gen class size} for $n\leq l$ to be 
	\begin{equation}
		F_{\text{cls}}(N,l)  = \frac{| \{ n \leq l\, \text{~s.t. $\mathbb{E}(B_{N,n}, \mathcal{C}_{\text{gens}}) > n$}\} |}{| \{ n \leq l \, \text{~s.t. $\mathbb{E}(B_{N,n}, \mathcal{C}_{\text{gens}}) < n$}\} |} \, . \label{def: eq F class}
	\end{equation}
	The corresponding quantity for fusion algebras is given by
	\begin{equation}
		F_{\text{fus}}(N, l)  = \frac{| \{ n \leq l \, \text{~s.t. $\mathbb{E}(B_{N,n}, \mathcal{R}_{\text{gens}}) > n$}\} |}{ | \{ n \leq l \, \text{~s.t. $\mathbb{E}(B_{N,n}, \mathcal{R}_{\text{gens}}) < n$}\} |}
		\, .\label{def: eq F fusion}
	\end{equation}
	This is the cumulative ratio of number of points above and below the line $y(n) = n$ in Figure \ref{fig: avg gen dim} for $n \leq l$.
	\label{def: ratio}
\end{definition}
These cumulative ratios are plotted in Figures \ref{fig: avg gen class size ratio} and \ref{fig: avg gen dim ratio}.
Let us now summarize our observation as follows.
\begin{proposition}
	\label{prop: ratio}
	The cumulative ratio of points above and below the line $y(n)=n$ in Figure \ref{fig: avg gen class size} has the following properties
	\begin{equation}
		 \begin{aligned}
			F_{\text{cls}}(30,l) &> 1, \quad \forall l  \\
			F_{\text{cls}}(30, \infty) &= \frac{198}{193} \approx 1.026, \quad \forall l  \\
		 \end{aligned}
	\end{equation}
	where $F_{\text{cls}}(30, \infty)$ is the minimum.
	The cumulative ratio of points above and below the line $y(n)=n$ in Figure \ref{fig: avg gen dim} has the following properties
	\begin{equation}
		 \begin{aligned}
			F_{\text{fus}}(30, l) &> 0.8 , \quad \forall l  \\
			F_{\text{fus}}(30, \infty) &= \frac{179}{212} \approx 0.844, \quad \forall l  \\
		 \end{aligned}
	\end{equation}
	where $F_{\text{fus}}(30, \infty)$ is the minimum.
\end{proposition}

\begin{figure}
	\centering
	\begin{subfigure}[t]{0.8\textwidth}
		\caption{}
		\vspace{0pt}
		\includegraphics[width=\textwidth]{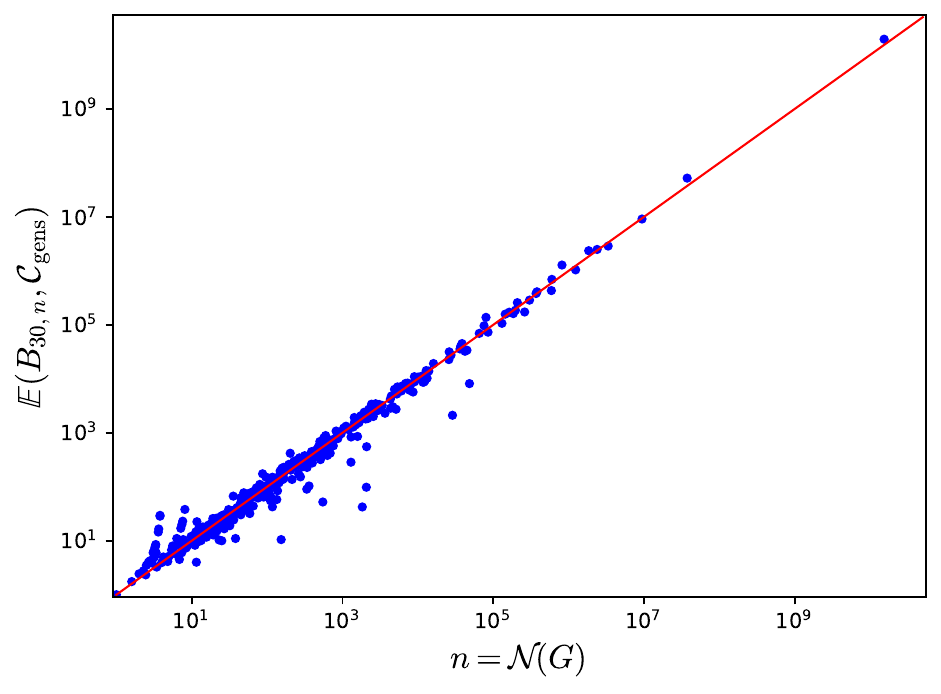}
	\label{fig: avg gen class size}
	\end{subfigure}
	\\
	\begin{subfigure}[t]{0.8\textwidth}
		\caption{}
		\vspace{0pt}
		\includegraphics[width=\textwidth]{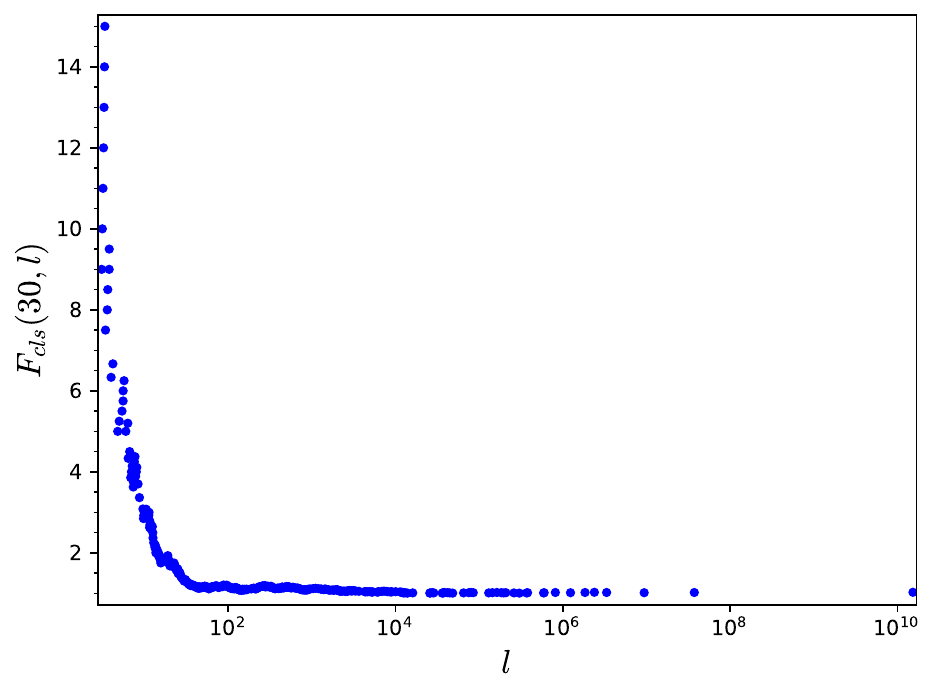}
	\label{fig: avg gen class size ratio}
	\end{subfigure}
	\caption{
		(a)
		The blue points show the average conjugacy class size of generating sets for the class algebra, $\mathbb{E}(E_{30,n}, \mathcal{C}_{\text{gens}})$, as a function of $n=\mathcal{N}(G)$.
		The red line is $y(n)=n$.
		The number of points above the red line is $198$ while the number of points below the red line is $193$. The number of points on the red line is $3$.
		(b)
		This plot displays $F_{\text{cls}}(30,l)$, the cumulative ratio of number of points above and below the line $y(n)=n$ in (a) from $n=1$ to $n=l$.
		}
	\label{fig: class size}
\end{figure}

\begin{figure}
	\centering
	\begin{subfigure}[t]{0.8\textwidth}
		\caption{}
		\vspace{0pt}
		\includegraphics[width=\textwidth]{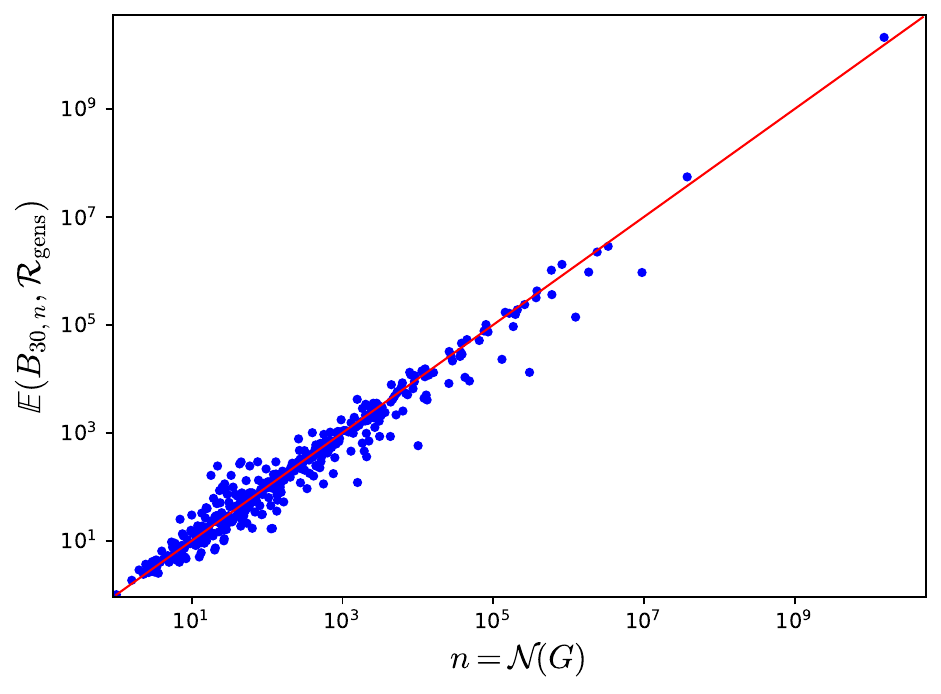}
	\label{fig: avg gen dim}
	\end{subfigure}
	\hfill
	\begin{subfigure}[t]{0.8\textwidth}
		\caption{}
		\vspace{0pt}
		\includegraphics[width=\textwidth]{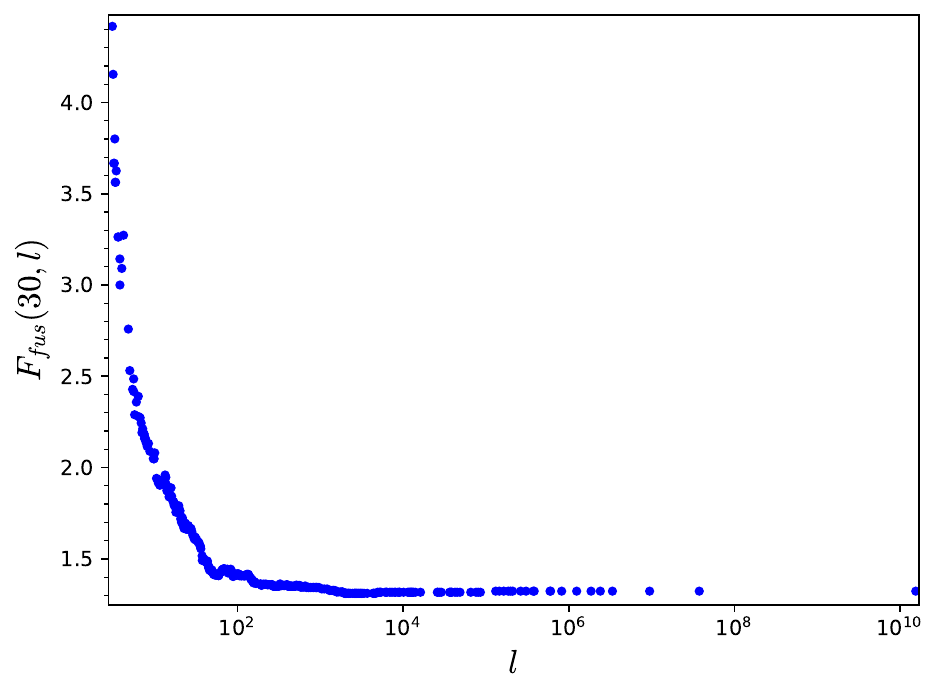}
	\label{fig: avg gen dim ratio}
	\end{subfigure}
	\caption{
		(a)
		The blue points show the average squared dimension of generating sets for the fusion algebra, $\mathbb{E}(E_{30,n}, \mathcal{R}_{\text{gens}})$, as a function of $n=\mathcal{N}(G)$.
		The red line is $y(n)=n$.
		The number of points above the red line is $179$ while the number of points below the red line is $212$ and the number of points on the red line is $3$.
		(b)
		This plot displays $F_{\text{fus}}(30,l)$, the cumulative ratio of number of points above and below the line $y(n)=n$ in (a) from $n=1$ to $n=l$.
	}
	\label{fig: dim squared}
\end{figure}

Based on Proposition \ref{prop: ratio}, we make the following conjectures for the class algebra.
\begin{conjecture}[Average class size (Strong version)]
	\label{conj: class size strong}
	The cumulative ratio $F_{\text{cls}}(N, l)$ is bounded by $1$.
	That is,
	\begin{equation}
		F_{\text{cls}}(N, l)  \geq  1
	\end{equation}
	for all $l$ and $N$.
\end{conjecture}
For some values of $N$ and $l$, the denominator in equation \eqref{def: eq F class} is zero.
We interpret this as $F_{\text{cls}}(N, l) > 1$.
One can also consider a weaker but also interesting version of the above conjecture where we take $N$ to be large.
We split the conjecture into two parts.
\begin{conjecture}[Average class size (Weak version (a))]
	As we include very large character tables, the cumulative ratio $F_{\text{cls}}(N, l)$ converges. 
	That is,
	\begin{equation}
		\lim_{N \rightarrow \infty} F_{\text{cls}}(N, l) 
	\end{equation}
	exists for every $l$.
\end{conjecture}

If the above conjecture is true, we can also consider the following conjecture.
\begin{conjecture}[Average class size (Weak version (b))]
	As we include very large character tables, the cumulative ratio $F_{\text{cls}}(N, l)$ is bounded by $1$.
	That is,
	\begin{equation}
		\lim_{N \rightarrow \infty} F_{\text{cls}}(N, l)  \geq  1
	\end{equation}
	for every $l$.
\end{conjecture}

In the case of fusion algebras, Proposition \ref{prop: ratio} and Figure \ref{fig: dim squared} do not give compelling evidence that the cumulative ratio $F_{\text{fus}}(N,l)$ converges to a nice value.
But it would be interesting to know if any lower bounds on $F_{\text{fus}}(N,l)$ exist.
We pose the following open question.
\begin{open}
	\hfill
	\begin{enumerate}
		\item Does there exist a positive number $F^{*}_{\text{fus}}$ such that
		\begin{equation}
			F_{\text{fus}}(N,l) \geq F^{*}_{\text{fus}}
		\end{equation}
		for all $l$ and $N$?
		\item Does the limit
		\begin{equation}
			\lim_{N \rightarrow \infty} F_{\text{fus}}(N, l) 
		\end{equation}
		exists for every $l$?
		\item If the answer to the above question is yes, does there also exist a positive number $F^{**}_{\text{fus}}$ such that
		\begin{equation}
			\lim_{N \rightarrow \infty} F_{\text{fus}}(N, l) \geq F^{**}_{\text{fus}}
		\end{equation}
		for all $l$?
	\end{enumerate}
\end{open}

Based on the data presented in this subsection, the average class size of generating sets of the class algebra has interesting properties.
In particular, it seems that the cumulative ratio $F_{\text{cls}}(N,l)$ is bounded by one, and approaches one in the limit where many character tables are included.

\section{Conclusions}

In this paper, we introduced a collection of  mathematical features of finite group  character tables inspired by measures of quantum and classical complexity in algorithms \cite{JosephSRProj}  related to AdS/CFT and topological quantum field theory.
Algebraically, the mathematical features amount to studying properties of minimal generating subsets of the class algebra and the fusion algebra of a finite group $G$.
The features are related by row-column duality in the sense of \cite{RowCol}: the first set involves the class algebra and conjugacy classes; the second set involves the fusion algebra and irreducible representations.

We then studied regularities that appear when taking statistical averages of the features over finite groups.
The class algebra circle generator complexity in equation \eqref{eq: complexity boundary class} and the fusion algebra  circle generator complexity in equation \eqref{eq: complexity boundary fusion} were compared in Figure \ref{fig: no gens boundary}.
We observed that the class algebra circle generator complexity tends to be smaller than the fusion algebra circle generator complexity,
\begin{equation*}
	\mathbb{E}(K_n, \brankC) < \mathbb{E}(K_n, \brankF) 
\end{equation*}
as formalized in Proposition \ref{prop: boundary complexity}.
We conjecture that this observation is robust in the sense of Conjectures \ref{conj: gen strong} and \ref{conj: gen weak}.

The average conjugacy class size over generating subsets and the average squared dimension over generating subsets were compared in Figures \ref{fig: class size} and \ref{fig: dim squared}.
The first part of Proposition \ref{prop: ratio} states that
\begin{equation*}
		F_{\text{cls}}(30,l) > 1, \quad \forall l  \,.
\end{equation*}
This formalised the observation that the average class size of minimal generating sets of the class algebra tends to be larger than the total average class size of finite groups, for groups with up to $N=30$ conjugacy classes.
In Conjecture \ref{conj: class size strong}, we proposed that this is true for general $N$. 
The second part of Proposition \ref{prop: ratio} states that 
\begin{equation*}
	F_{\text{fus}}(30, l) > 0.8 , \quad \forall l \, . 
\end{equation*}
That is, the average squared dimension of minimal generating sets of the fusion algebra tends to be smaller than the total average squared dimension, for groups with number of irreducible representations up to $N=30$.

This work opens up a number of interesting avenues for future research.
Finding ways to prove  the proposed conjectures is evidently a fascinating challenge, since these conjectures involve statistical averages involving groups of arbitrarily large size.
Equally, finding counter-examples or further computational evidence can guide theoretical approaches to proving or reformulating the conjectures. The estimates of the complexity of projector detection tasks for centres of symmetric group algebras  in \cite{JosephSRProj}, using quantum phase estimation, employed generator complexities of the kind we studied in detail here, along with estimates of the range of eigenvalues of the elements of the generators in the minimal generating sets. The analogous calculations for general groups, combining the generator complexities with eigenvalue ranges to be determined, is an interesting avenue for the future. 
This  would progress the present work closer to complexities of  quantum algorithms which can be physically implemented.
While the symmetric group projector detection task was motivated, through mechanisms involving Schur-Weyl duality to AdS/CFT, finding a holographic interpretation for the case of general groups, either in AdS/CFT or mathematical models of gauge-string duality based on matrix models,  is an interesting challenge. We hope to address some of  these challenges  using a combination of methods including a data scientific point of view, with the use of machine learning techniques, in the near future.

\section*{Acknowledgements}
This work was initiated as a collaboration in the context of a Royal Society International Exchanges Award, IES$\setminus$R2$\setminus$222073,  held by SR and RKS.
SR is supported by the Science and Technology Facilities Council
(STFC) Consolidated Grant ST/T000686/1 
``Amplitudes, strings and duality''. SR acknowledges a  Visiting Professorship at the Dublin Institute for Advanced Studies held during the progress of this work. SR also gratefully acknowledges a visit to the Perimeter Institute: this research was supported in part by Perimeter Institute for Theoretical Physics. Research at Perimeter Institute is supported by the Government of Canada through the Department of Innovation, Science, and Economic Development, and by the Province of Ontario through the Ministry of Colleges and Universities.
The work of AP was partly funded by the Deutsche Forschungsgemeinschaft (DFG) grant SFB 1283/2 2021 E317210226.
RKS was supported by a Basic Research Grant of the National
Research Foundation of Korea (NRF2022R1F1A1073128) during the course of this project.
He is also supported by an Outstanding Young Scientist Grant of the National
Research Foundation of Korea,
and partly supported by the BK21 Program
(`Next Generation Education Program for Mathematical Sciences’, 4299990414089)
funded by the Ministry of Education in Korea and the National Research Foundation
of Korea. We  thank Joseph Ben Geloun and Rajath Radhakrishnan for discussions related to this work.

\section{Appendix}

\subsection{Algorithms}

The minimum number of columns (each corresponding to a conjugacy class of a group $G$)
needed to distinguish all rows (each representing an irreducible representation of $G$)
in the character table 
is precisely the circle generator complexity of the class algebra of $G$.
This complexity is 
denoted by
$N_{\text{cls}}(G) \equiv N(\mathcal{Z}(\mathbb{C}[G]))$
in Sections \ref{sec: review tqft} and \ref{sec: class complexity}.
Table \ref{tab:alg01} summarizes the algorithm used to identify this
minimal set of columns in the character table. 

\begin{table}[H]
\centering
\resizebox{0.8\hsize}{!}{
\begin{tabular}{l}
\hline
\textbf{Algorithm for \texttt{minimum\_columns\_to\_distinguish\_rows}} \\
\hline
\\[-6pt]
\textbf{FUNCTION} \texttt{minimum\_columns\_to\_distinguish\_rows}(\textit{matrix}): \\
\quad \textbf{let} $n \leftarrow$ \textit{number of columns in matrix} \\
\quad \textbf{let} \textit{rows} $\leftarrow \textit{matrix}$ \\
\\[-6pt]
\quad \textbf{for} $k$ \textbf{from} 1 \textbf{to} $n$: \\
\quad\quad \textbf{for each} \textit{column\_indices} \textbf{in} \textbf{COMBINATIONS\_OF}($k$, $\{0 \ldots n-1\}$): \\
\quad\quad\quad \textbf{let} \textit{projections} $\leftarrow$ \textit{empty list} \\
\quad\quad\quad \textbf{for each} \textit{row} \textbf{in} \textit{rows}: \\
\quad\quad\quad\quad \textbf{let} \textit{projection} $\leftarrow$ \textit{empty list} \\
\quad\quad\quad\quad \textbf{for each} \textit{column\_index} \textbf{in} \textit{column\_indices}: \\
\quad\quad\quad\quad\quad \textbf{append} \textit{row}[\textit{column\_index}] \textbf{to} \textit{projection} \\
\quad\quad\quad\quad \textbf{append} \textbf{TUPLE}(\textit{projection}) \textbf{to} \textit{projections} \\
\\[-6pt]
\quad\quad\quad \textbf{if} \textbf{SIZE}(\textbf{UNIQUE\_SET}(\textit{projections})) = $n$: \\
\quad\quad\quad\quad \textbf{let} \textit{columns\_needed} $\leftarrow$ [ \textbf{TRANSPOSE}(\textit{matrix})[c] \\
\quad\quad\quad\quad\quad\quad\quad\quad\quad\quad\textbf{for} c \textbf{in} \textit{column\_indices} ] \\
\quad\quad\quad\quad \textbf{return} $(k,\ \textit{column\_indices},\ \textit{columns\_needed})$ \\
\\[-6pt]
\quad \textbf{return} $(n,\ \{0, 1, \ldots, n-1\},\ \textbf{TRANSPOSE}(\textit{matrix}))$ \\
\\[-6pt]
\hline
\end{tabular}
}
\caption{Algorithm for distinguishing rows with a minimal set of columns in a $n\times n$ matrix.}
\label{tab:alg01}
\end{table}

Similarly, we note that the minimum number of rows (each corresponding to an irreducible representation of a group $G$)
needed to distinguish all columns (each representing a conjugacy class of $G$)
in the character table 
is precisely the circle generator complexity of the fusion algebra of $G$.
This complexity is 
denoted by
$N_{\text{fus}}(G) \equiv N(R(G))$
in Sections \ref{sec: review tqft} and \ref{sec: class complexity}.
Table \ref{tab:alg02} summarizes the algorithm used to identify this
minimal set of rows in the character table. 

\begin{table}[H]
\centering
\resizebox{0.8\hsize}{!}{
\begin{tabular}{l}
\hline
\textbf{Algorithm for \texttt{minimum\_rows\_to\_distinguish\_columns}} \\
\hline
\\[-6pt]
\textbf{FUNCTION} \texttt{minimum\_rows\_to\_distinguish\_columns}(\textit{matrix}): \\
\quad \textbf{let} $n \leftarrow$ \textit{number of rows in matrix} \\
\quad \textbf{let} \textit{columns} $\leftarrow$ \textbf{TRANSPOSE}(\textit{matrix}) \\
\\[-6pt]
\quad \textbf{for} $k$ \textbf{from} 1 \textbf{to} $n$: \\
\quad\quad \textbf{for each} \textit{row\_indices} \textbf{in} \textbf{COMBINATIONS\_OF}($k$, $\{0 \ldots n-1\}$): \\
\quad\quad\quad \textbf{let} \textit{projections} $\leftarrow$ \textit{empty list} \\
\quad\quad\quad \textbf{for each} \textit{column} \textbf{in} \textit{columns}: \\
\quad\quad\quad\quad \textbf{let} \textit{projection} $\leftarrow$ \textit{empty list} \\
\quad\quad\quad\quad \textbf{for each} \textit{row\_index} \textbf{in} \textit{row\_indices}: \\
\quad\quad\quad\quad\quad \textbf{append} \textit{column}[\textit{row\_index}] \textbf{to} \textit{projection} \\
\quad\quad\quad\quad \textbf{append} \textit{projection} \textbf{to} \textit{projections} \\
\\[-6pt]
\quad\quad\quad \textbf{if} \textbf{SIZE}(\textbf{UNIQUE\_SET}(\textit{projections})) = \textbf{SIZE}(\textit{columns}): \\
\quad\quad\quad\quad \textbf{let} \textit{rows\_needed} $\leftarrow$ [\textit{matrix}[\textit{r}] \textbf{for} \textit{r} \textbf{in} \textit{row\_indices}] \\
\quad\quad\quad\quad \textbf{return} $(k, \textit{row\_indices}, \textit{rows\_needed})$ \\
\\[-6pt]
\quad \textbf{return} $(n,\ \{0, 1, \ldots, n-1\},\ \textit{matrix})$ \\
\\[-6pt]
\hline
\end{tabular}
}
\caption{Algorithm for distinguishing columns with a minimal set of rows in a $n\times n$ matrix.}
\label{tab:alg02}
\end{table}

\subsection{Proof of Theorem \ref{thm: generators}}
\label{apx: generators}
In this appendix, we prove Theorem \ref{thm: generators}, using a generalisation of the arguments in Section 3.4 of \cite{KempRam} (see also \cite[Section 3.1]{SRES}).

We start by studying the space  generated by a single element $t_1 \in A$.
Using 
\begin{equation}
	t_1 = \sum_{a=1}^K \widehat{\chi}^a(t_1) P_a\, ,
\end{equation}
and $P_a P_b = \delta_{ab} P_a$, we have
\begin{equation}
	t_1^n = \sum_{a} [\widehat{\chi}^a(t_1)]^n P_a \, .\label{eq: t1 to n}
\end{equation}
We split the sum over $a$ into level sets of $\widehat{\chi}^a(t_1)$.

For this, it is useful to consider the set of all distinct values taken by the character $\widehat{\chi}^a(t_1)$ as we vary $a$,
\begin{equation}
	\mathcal{S}(t_1) = \{\widehat{\chi}^a(t_1) \}_{a=1}^K \, .
\end{equation}
For a fixed value $\xi_1 \in \mathcal{S}(t_1)$, we define the subset $O(\xi_1)$ of $\{1,\dots,K\}$ for which the character agrees with $\xi_1$,
\begin{equation}
	O(\xi_1) = \{ b \in \{1,\dots,K\} \, : \, \widehat{\chi}^b(t_1) = \xi_1 \} \, .
\end{equation}
Note that these subsets $O(\xi_1)$ partition $\{1,\dots,K\}$ into disjoint subsets such that
\begin{equation}
	\{ 1 ,\dots, K \} = \bigcup_{ \xi_1 \in \mathcal{S}(t_1)} O(\xi_1) \, .
\end{equation}

We re-write equation \eqref{eq: t1 to n} as follows
\begin{equation}
	t_1^n = \sum_{ \xi_1 \in \mathcal{S}(t_1) } \xi_1^n \sum_{b \in O(\xi_1)} P_b \, , \label{eq: t1 to n 2}
\end{equation}
and define
\begin{equation}
	P(\xi_1) = \sum_{b \in O(\xi_1)} P_b, \quad V_{\xi_1 n} = \xi_1^n \, .
\end{equation}
Then equation \eqref{eq: t1 to n 2} can be understood as a matrix equation as follows
\begin{equation}
	t_1^n = \sum_{ \xi_1 \in \mathcal{S}(t_1)} V_{\xi_1 n } P(\xi_1) \, .
\end{equation}
By construction, $V$ is a Vandermonde matrix with distinct entries.
We restrict $n$ to be in the set $\{0,\dots, |\mathcal{S}(t_1)| - 1 \}$ making $V$ invertible.
Let $W$ be the inverse of $V$, then we have
\begin{equation}
	\sum_{b \in O(\xi_1)} P_b	 = P(\xi_1) = \sum_{n=0}^{|\mathcal{S}(t_1)|-1} t_1^n W_{n \xi_1} \, . \label{eq: pb vandermonde inverse}
\end{equation}

The right-hand side can be written explicitly, as we now explain.
For this, we introduce the elementary symmetric functions of a set of numbers $\{x_a\}_{a=1}^m$, denoted $e_i(\{x_a\}_{a=1}^m)$.
They appear in the expansion of the product
\begin{equation}
	\prod_{a=1}^m (\lambda - x_a) = \sum_{i=0}^{m} \lambda^{i}(-1)^{m-i} e_{m-i}(\{x_a\}_{a=1}^m) \, . \label{eq: Vieta}
\end{equation}
In terms of elementary symmetric functions, we have (see Section 1.2.3 exercise 40 of \cite{TAOCP1})
\begin{equation}
	W_{n \xi_1} = (-1)^{|\mathcal{S}(t_1)|-1-n} \frac{e_{|\mathcal{S}(t_1)|-1-n}(\mathcal{S}(t_1)\backslash \{\xi_1\})}{\prod_{\xi \in \mathcal{S}(t_1)\backslash \{\xi_1\} }(\xi_1 - \xi)} \, .
\end{equation}
Plugging this into \eqref{eq: pb vandermonde inverse} gives
\begin{equation}
	P(\xi_1) = \sum_{n=0}^{|\mathcal{S}(t_1)|-1} t_1^n (-1)^{|\mathcal{S}(t_1)|-1-n}\frac{e_{|\mathcal{S}(t_1)|-1-n}(\mathcal{S}(t_1)\backslash \{\xi_1\})}{\prod_{\xi \in \mathcal{S}(t_1)\backslash \{\xi_1\} }(\xi_1 - \xi)} \, .
\end{equation}
Using equation \eqref{eq: Vieta} again, we get
\begin{equation}
	P(\xi_1) = \prod_{\xi \in \mathcal{S}(t_1)\backslash \{\xi_1\}} \frac{(t_1 - \xi)}{(\xi_1 - \xi)} \, .
\end{equation}

In particular, the space generated by powers of $t_1$ has a basis
\begin{equation}
	\Span( 1, t_1, t_1^2, \dots, ) = \Span( P(\xi_1) \, : \, \xi_1 \in \mathcal{S}(t_1) ) \, .
\end{equation}
Before we move on to the case of two generators, note that if the value of $\widehat{\chi}^a(t_1)$ distinguishes all $a$, each level set $O(\xi_1)=\{b\}$ contains a single $b$.
It follows that
\begin{equation}
	P(\xi_1) = P_b \, .
\end{equation}
In other words, $t_1$ generates the full algebra $A$.

Now consider two elements $t_1, t_2$ and the subalgebras they generate separately.
The first element generates a space with basis elements
\begin{equation}
	P(\xi_1), \quad \xi_1 \in \mathcal{S}(t_1) \, .
\end{equation}
The second element generates a space with basis elements
\begin{equation}
	P(\xi_2), \quad \xi_2 \in \mathcal{S}(t_2) \, .
\end{equation}
The product of such basis elements can be written in terms of intersections of level sets and is given by
\begin{equation}
	P(\xi_1) P(\xi_2) = \sum_{\substack{a \in O(\xi_1) \\ b \in O(\xi_2) }} P_a P_b = \sum_{\substack{c \in O(\xi_1) \cap O(\xi_2) }} P_c \, .
\end{equation}
The last equality follows from $P_a P_b = \delta_{ab} P_b$ and the sum is over all $c \in \{1, \dots, K\}$ such that
\begin{equation}
	( \widehat{\chi}^c(t_1), \widehat{\chi}^c(t_2) ) = ( \xi_1 , \xi_2) \, .
\end{equation}
Now, if the list $(\widehat{\chi}^a(t_1), \widehat{\chi}^a(t_2) )$ distinguishes all $a$, then the intersection $O(\xi_1) \cap O(\xi_2) = \{b\} $ contains a single element and therefore
\begin{equation}
	P(\xi_1) P(\xi_2) =  P_b
\end{equation}
and $t_1,t_2$ generate the full algebra.

The generalisation to $l$ generators is straight-forward. Every $P_b$ has an expression as a product of the form
\begin{equation}
	P_b =  \prod_{i=1}^l P(\xi_i) \, ,
\end{equation}
where $\xi_i \in \mathcal{S}(t_i)$ are the character values that uniquely determine $b$.
That is,
\begin{equation}
	\bigcap_{i=1}^l O(\xi_i) = \{b\} \, .
\end{equation}

\printbibliography
\end{document}